\documentclass[12pt,reqno]{amsart}
\usepackage[hypertex]{hyperref}
\usepackage[final]{graphicx}
\usepackage{amsfonts}
\usepackage{amsmath}
\usepackage{amssymb}
\usepackage{amsthm}

\newtheorem{theorem}{Theorem}[section]
\newtheorem{lemma}[theorem]{Lemma}
\newtheorem{corollary}[theorem]{Corollary}
\newtheorem{proposition}[theorem]{Proposition}

\newtheorem{definition}[theorem]{Definition}
\newtheorem{conjecture}[theorem]{Conjecture}

\newcommand{\bra}{\langle}
\newcommand{\ket}{\rangle}
\newcommand{\upchi}{\raise1pt\hbox{$\chi$}}
\newcommand{\R}{{\mathord{\mathbb R}}}

\newcommand{\Z}{{\mathord{\mathbb Z}}}

\begin{document}

\title[The Random Displacement Model]{Spectral Properties of the Discrete Random Displacement Model}

\author{Roger Nichols}

\address{Department of Mathematics, University of Alabama at Birmingham, Birmingham, AL~35294, USA, and, Department of Mathematics,
University of Missouri, Columbia, MO 65211, USA (starting September
1, 2010)}

\email{rnich02@math.uab.edu}

\author{G\"unter Stolz}

\address{Department of Mathematics, University of Alabama at Birmingham, Birmingham, AL~35294, USA}

\email{stolz@math.uab.edu}

\thanks{G.\ S.\ was supported in part by NSF grant DMS-0653374.}

\maketitle

\begin{abstract}

 We investigate spectral
properties of a discrete random displacement model, a Schr\"odinger
operator on $\ell^2(\Z^d)$ with potential generated by randomly
displacing finitely supported single-site terms from the points of a
sublattice of $\Z^d$. In particular, we characterize the upper and
lower edges of the almost sure spectrum. For a one-dimensional model
with Bernoulli distributed displacements, we can show that the
integrated density of states has a $1/\log^2$-singularity at external
as well as internal band edges.

\end{abstract}

\section{Introduction}

While the Anderson model is by far the most studied model of a
random Schr\"odinger operator and still poses open problems, recent
years have seen an increased interest in other types of random
operators such as the Poisson model, the random displacement model,
or Schr\"odinger operators involving random magnetic fields. Some recent references are \cite{GHK1, GHK2, BLS1, BLS2, Bourgain}. While there are good physical reasons to look
at these models, an equally strong mathematical motivation for their
investigation stems from challenges arising due to the lack of
monotonicity properties in these models. Such properties have been
heavily used in the theory of Anderson-type models. However, even in
the Anderson model, non-monotonicity issues arise if one considers
single-site terms which are not sign-definite, see, e.g.\, \cite{KN1, KN, Veselic}.

Absence of monotonicity requires new ideas, which, besides posing a
mathematical challenge, often require a better understanding of
physical mechanisms, typically in the form of a more subtle
interplay between kinetic and potential energy. In particular, this
has become apparent in the recent works \cite{BLS1, BLS2, KLNS} on
the continuum random displacement model (RDM), a random
Schr\"odinger operator of the form
\[ H_{\omega} = -\Delta + \sum_{n\in \Z^d} q(x-n-\omega_n)\]
in $L^2(\R^d)$ with single-site terms $q$ displaced randomly by
vectors $\omega = (\omega_n)_{n\in \Z^d}$ from the sites $n$ of the
lattice $\Z^d$. It was found that the lack of monotonicity can be
widely remedied by symmetry considerations as long as one assumes
corresponding symmetry properties of the single-site potential. This
has led to key insights for the RDM, such as a characterization of
the spectral minimum, properties of the integrated density of states
and a Wegner estimate at low energy, ultimately leading to a proof of localization
near the bottom of the spectrum for the continuum RDM in
\cite{KLNS}.

Our main goal here is to provide analogues of the results in
\cite{BLS1} and \cite{BLS2} for a {\it discrete} version of the RDM.
While we largely succeed in this attempt, in some instances we fall short of
carrying over results from the continuum, which is
mostly due to a well known problem for lattice operators, the lack
of unique continuation properties which are frequently used in
\cite{BLS1, BLS2}. In particular, this is the reason why a proof of
localization for models like the one considered here, just as for the discrete Anderson model with singularly distributed coupling constants, is still out of
reach, see the remarks in Section~\ref{sec:conclusion}.

However, this shortcoming does not affect the 1D case, where we
recover all the results from the continuum. In fact, for $d=1$ we
find new phenomena, not encountered in the continuum. In particular,
for a one-dimensional RDM with Bernoulli distributed random
displacements we find a gap in the almost sure spectrum and,
in the case of symmetric distribution, are able to investigate the
behavior of the integrated density of states (IDS) at the spectral minimum and maximum as well as at the edges of this gap. At all these edges the IDS has a $1/\log^2$-singularity and, in particular, is not H\"older continuous.

We mention that localization at all energies for the one-dimensional discrete RDM has been proven in \cite{DSS}, where the more general setting of {\it random word models} was considered. This is based on showing that the Lyapunov exponent is positive at all but an at most finite set of critical energies, which may give rise to quantum transport for wave packets with energy support close to the critical energies, while it does not inhibit spectral localization and also leads to dynamical localization away from the critical energies.

The remaining sections of this paper are structured as follows: In Section~\ref{sec:modelresults} we introduce the discrete displacement model and state all our results. Here we also formulate a discrete version of the fact that ``bubbles tend to the corners'', a result originally proven in \cite{BLS1} for the continuum RDM which provides a central tool for our work. Sections~\ref{sec:multiDresults}, \ref{sec:nonunique} and \ref{sec:Bernoulli} contain all proofs, with Section~\ref{sec:multiDresults} giving results which hold in arbitrary dimension and Sections~\ref{sec:nonunique} and \ref{sec:Bernoulli} presenting proofs of results for the one-dimensional model. In particular, we discuss the one-dimensional Bernoulli RDM in Section~\ref{sec:Bernoulli}. Section~\ref{sec:conclusion} contains concluding remarks including several open conjectures, partly based on numerical observations presented there.

\section{Model and Results} \label{sec:modelresults}

\subsection{Basics}

We will construct a random potential on $\Z^d$, $d\ge 1$, by randomly placing single-site terms supported in a rectangular box into translates of a larger box. The two basic boxes are
    \begin{displaymath}
      \Lambda:=\prod_{i=1}^d[1,M_i]\subset \mathbb{Z}^d \qquad \textrm{and} \qquad B:=\prod_{i=1}^d[1,b_i] \subset \mathbb{Z}^d
    \end{displaymath}
    where $1\leq b_i\leq M_i$ and $b_i,M_i\in \mathbb{N}$ are fixed for each $i\in \{1,\ldots, d\}$.  As single site potential we choose a function $q:\Z^d \to \R$ which is supported in $B$. We shall always assume the hypothesis:

\begin{quote}
($\mathbf{H1}$) The single-site potential $q$ is reflection symmetric in each variable in the sense that
\begin{displaymath}
q(x_1,\ldots,x_{i-1},x_i,x_{i+1},\ldots,x_d)=q(x_1,\ldots,x_{i-1},b_i-x_i+1,x_{i+1},\ldots,x_d)
\end{displaymath}
for each $x=(x_j)_{j=1}^d$ and all $i\in \{1,\ldots,d\}$.
\end{quote}

We denote the translate of $q$ by $a\in \Z^d$ as $q_a$, that is $q_a(n)=q(n-a)$, $n\in \Z^d$. We will generally require that the support of $q_a$ remains in $\Lambda$, meaning
\[ a \in \Delta := \prod_{i=1}^d [0,M_i-b_i].\]

As kinetic energy operator we choose the (negative) discrete Laplacian $h_0$ on $\ell^2(\mathbb{Z}^d)$, that is, for $u\in \ell^2(\mathbb{Z}^d)$,
\begin{displaymath}
(h_0u)(n)=-\sum_{\substack{m\in \mathbb{Z}^d \\ |m-n|=1}}u(m),
\end{displaymath}
where $|k|$ denotes the $1$-norm of a vector $k\in \mathbb{Z}^d$.

\subsection{Bubbles Tend to the Corners}

As in \cite{BLS1}, a key ingredient into our investigations of the RDM will be given by a property of the single-site operators $h_{0,\Lambda}^N+q_a$ on $\ell^2(\Lambda)$, where $h_{0,\Lambda}^N$ denotes the (discrete) Neumann Laplacian on $\Lambda$ (see Section~\ref{sec:NeuLap} below for a precise definition of the Neumann Laplacian). Define
\begin{equation}
E_0(a):=\min \sigma(h_{0,\Lambda}^N+q_a), \ a\in \Delta.
\end{equation}
Hypothesis $(\mathbf{H1})$ implies that $E_0(\cdot)$ is reflection symmetric on $\Delta$, i.e.
\begin{equation}
E_0(a_1,\ldots,a_{i-1},a_i,a_{i+1},\ldots, a_d)=E_0(a_1,\ldots,a_{i-1},M_i-b_i-a_i,a_{i+1},\ldots,a_d)
\end{equation}
for all $i\in \{1,\ldots,d\}$ and $a=(a_i)_{i=1}^d\in \Delta$.  Thus, $E_0(\cdot)$ is determined by $a\in \prod_{i=1}^d[r_i,M_i-b_i]$, where $r_i$ is the least integer greater than or equal to $(M_i-b_i)/2$.

\begin{theorem}\label{thm: bubbles}
Fix $i\in \{1,\ldots,d\}$ and $a_j\in [r_j,M_j-b_j]$, $j\in \{1,\ldots,d\}\setminus \{i\}$ and suppose $(\mathbf{H1})$ holds.  If either
\begin{enumerate}
\item[(i)] $q\neq 0$ is sign-definite, or
\item[(ii)] $d=1$ and $E_0(a)\neq-2$ for at least one $a\in \Delta$,
\end{enumerate}
then $E_0(a)$ is strictly decreasing as a function of $a_i$ on $[r_i,M_i-b_i]$.
\end{theorem}

This holds in each variable, meaning, in particular, that $E_0(\cdot)$ attains strict minima in the $2^d$ corners of $\Delta$. The number $-2$ appears in (ii) as the spectral minimum of $h_{0,\Lambda}^N$ in $d=1$. See Section~\ref{sec:discussion} for a discussion of the relevance of assumptions (i), (ii) as well as for a comparison with the corresponding result in the continuum proven in \cite{BLS1}.

\subsection{The Discrete Displacement Model}

We now construct potentials on $\Z^d$ by tiling $\Z^d$ with translates of $\Lambda$ and placing one copy of $q$ into each tile. More precisely, if
\[ \Omega := \Delta^{\Z^d}\]
is the set of all possible displacement configurations and $\omega = (\omega_k)_{k\in \Z^d} \subset \Omega$, then $V_{\omega}:\mathbb{Z}^d\rightarrow \mathbb{R}$ is defined by
\begin{equation}\label{eqn: random potential}
V_{\omega}:=\sum_{k\in \mathbb{Z}^d}q_{kM+\omega_k},
\end{equation}
where $kM:=(k_iM_i)_{i=1}^d$ for $k=(k_i)_{i=1}^d$ and $M=(M_i)_{i=1}^d$.  For each $\omega \in \Omega$,
\begin{equation}\label{eqn: SO}
h_{\omega}:=h_0+V_{\omega}
\end{equation}
defines a bounded self-adjoint operator on $\ell^2(\mathbb{Z}^d)$.

The main challenge in understanding the displacement model lies in its non-monotonicity (in form sense) in the displacement parameters $\omega_n$. For example, it is not immediately clear which configurations minimize (or maximize) the spectrum of $h_{\omega}$. It is our first goal to answer this question.

The family $h_{\omega}$ is uniformly bounded in $\omega \in \Omega$, $\|h_{\omega}\|\leq 2d + \sup_{n\in B}|q(n)|$. Therefore,
\[
E_{min}:=\inf_{\omega\in \Omega} \min \sigma(h_{\omega})>-\infty, \quad
E_{max}:=\sup_{\omega\in \Omega} \max \sigma(h_{\omega})<\infty.
\]

We are concerned with the existence of displacement configurations $\omega$ which are spectrally minimizing in the sense that
\[ \min \sigma(h_{\omega}) = E_{min}\]
or spectrally maximizing,
\[ \max \sigma(h_{\omega}) = E_{max}.\]

Under suitable additional assumptions, the answer to both questions will be given by the configuration $\omega^* =(\omega^*_k)_{k\in \mathbb{Z}^d}\in \Omega$ defined by
\begin{displaymath}
(\omega^{\ast}_k)_i=\left\{ \begin{array}{cc}
0 & \textrm{if} \ k_i \ \textrm{is even}\\
M_i-b_i & \textrm{if} \ k_i \ \textrm{is odd}
\end{array}\right.
\end{displaymath}
for $i\in \{1,\ldots,d\}$. This is the periodic configuration in which clusters of $2^d$ single-site terms are placed into adjacent corners of their supporting tiles, see Figure~\ref{clusteringpic}.

\begin{theorem}\label{thm: omega star}
Suppose $(\mathbf{H1})$ holds.

(a) If either,
\begin{itemize}
\item[(i)] $d\geq2$ and $q$ is sign-definite, or
\item[(ii)] $d=1$ and $E_0(a)\not= -2$ for at least one $a\in \Delta$,
\end{itemize}
 then the displacement configuration $\omega^*$ is spectrally minimizing.

(b) Let $\tilde{E}_0(a) := \min \sigma(h_{0,\Lambda}^N-q_a)$. If either (i) $d\geq2$ and $q$ is sign-definite, or (ii) $d=1$ and $\tilde{E}_0(a)\not= -2$ for at least one $a\in \Delta$, then $\omega^*$ is spectrally maximizing.

\end{theorem}

    \begin{figure}
      \centering
      \includegraphics[width=.55\textwidth]{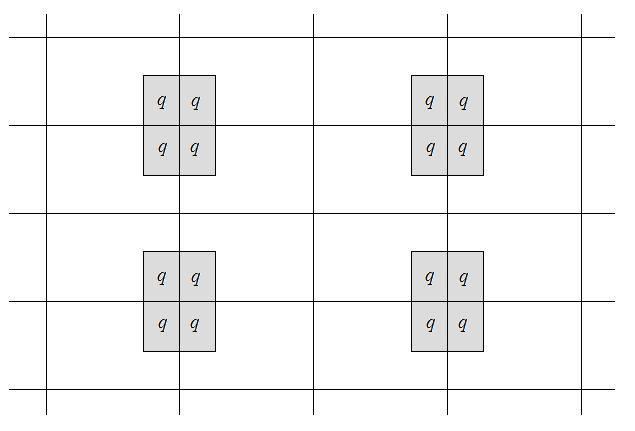}
      \caption[Extremal Potential, $V_{\omega^{\ast}}$, $d=2$]{The potential $V_{\omega^{\ast}}$ corresponding to the extremal configuration $\omega^{\ast}$ in $d=2$ depicting the clustering of neighboring single-site potentials.  Shaded areas represent the support of copies of the single-site potential $q$.}
      \label{clusteringpic}
    \end{figure}

\subsection{The Random Displacement Model}

The displacement model is called \textit{random} if $\omega=(\omega_k)_{k\in \mathbb{Z}^d}\in \Omega$ is a collection of i.i.d. $\Delta$-valued random variables.  That is, the $\omega_k$, $k\in \Z^d$, are independent, and
\begin{displaymath}
\mathbb{P}(\omega_k=n)=\mu(\{n\}),
\end{displaymath}
for a given (fixed) distribution $\mu$ on $\Delta$, meaning $\mu(\{n\})\geq 0$ for all $n\in \Delta$ and $\sum_{n\in \Delta}\mu(\{n\})=1$. $\mathbb{P}$ is realized as the infinite product measure, indexed by $k\in \Z^d$, of the measures $\mu$ on $\Delta$. By $\mathbb{E}$ we denote the expectation with respect to $\mathbb{P}$.

In this case, the random displacement model is ergodic with respect to shifts $(T_j \omega)_n = \omega_{n+j}$ on $\Omega$ and therefore, in particular, has deterministic spectrum.  Thus, there exists a closed set $\Sigma\subset \mathbb{R}$ such that $\sigma(h_{\omega})=\Sigma$ almost surely.  In fact, we can characterize $\Sigma$ in terms of the following ``periodic support theorem''. This is essentially a special case of results presented in \cite{KirschSO}, needing only slight adaptations due to the fact that our model is ergodic with respect to a sublattice of $\Z^d$.

\begin{theorem}\label{thm: deterministic spectrum}
For the random displacement model, one may take
\begin{equation}\label{eqn: deterministic spectrum}
\Sigma=\overline{\bigcup_{\omega \in \mathcal{C}_{per}}\sigma(h_{\omega})},
\end{equation}
where
\begin{displaymath}
\mathcal{C}_{per}=\{\omega \in \Omega: \omega \ \textrm{is periodic}, \ \mu(\omega_k)>0 \ \forall k\in \mathbb{Z}^d\}.
\end{displaymath}
\end{theorem}

Here we call a displacement configuration $\omega$ {\it periodic} if $\omega_{n+jK} = \omega_n$ for all $n, j \in \Z^d$ and a fixed vector $K=(K_i)_{i=1}^d$ with positive integer components $K_i$.

A vector $n\in \Delta$ is called a corner of $\Delta$ if $n_i\in \{0,M_i-b_i\}$ for all $i\in \{1,\ldots,d\}$.  Throughout this paper, we make the following assumption on the distribution $\mu$:

\begin{quote}
($\mathbf{H2}$) The distribution $\mu$ satisfies $\mu(n)>0$ for all corners $n\in \Delta$.
\end{quote}

Applying (\ref{eqn: deterministic spectrum}), an easy consequence of Theorem~\ref{thm: omega star} is the following characterization of $\inf \Sigma$ and $\sup \Sigma$:

\begin{corollary}\label{cor: upper and lower edges}
If $(\mathbf{H2})$ and the assumptions of Theorem \ref{thm: omega star} (a) and (b), respectively, hold, then the upper and lower edges of the almost-sure spectrum $\Sigma$ are characterized in terms of $\omega^{\ast}$ by
\begin{eqnarray}
\inf \Sigma&=&\min \sigma(h_{\omega^{\ast}})\nonumber\\
\sup \Sigma&=&\max \sigma(h_{\omega^{\ast}}).\nonumber
\end{eqnarray}
\end{corollary}

\subsection{A Uniqueness Result}

While $\omega^*$ is clearly not the unique spectral minimizer within {\it all} $\omega\in \Omega$ (in fact, we have $\inf \sigma(h_{\omega}) = \inf \Sigma$ for {\it almost every} $\omega \in \Omega$), it makes sense to ask if $\omega^*$ is the unique {\it periodic} minimizer. In dimension one this has a negative answer, but we are able to characterize {\it all periodic} minimizers. This will be used in the proof of our results on the integrated density of states described in the next subsection.

Let $d=1$ and $L\in \mathbb{N}$,  and $S_L$ denote the set of all $L$-periodic displacement configurations, i.e.
\begin{displaymath}
S_L:=\{\omega \in [0,M_1-b_1]^{\mathbb{Z}}:\omega_{k+L}=\omega_k \ \forall k\in \mathbb{Z}\}.
\end{displaymath}
For a displacement configuration $\omega \in [0,M_1-b_1]^{\mathbb{Z}}$ define the numbers
\begin{eqnarray}
n^0(\omega)&=&\# \{k: k\in \{1,\ldots,L\} \ \textrm{and} \ \omega_k=0\}\nonumber\\
n^1(\omega)&=&\# \{k: k\in \{1,\ldots,L\} \ \textrm{and} \ \omega_k=M_1-b_1\}.\nonumber
\end{eqnarray}

\begin{theorem}\label{thm: all 1d minimizers}
Assume that $\min \sigma(h_{0,\Lambda}^N+q_a)\neq -2$ for some $a\in [0,M_1-b_1]$.  Then $\omega \in S_L$ is spectrally minimizing if and only if $L$ is even and $n^0(\omega)=n^1(\omega)=L/2$.

If, on the other hand, $\min \sigma(h_{0,\Lambda}^N-q_a)\neq -2$ for some $a\in [0,M_1-b_1]$, then the same set of configurations characterizes all periodic spectral maximizers.

\end{theorem}

We discuss our expectation for uniqueness results in $d\ge 2$ in Section~\ref{sec:discussion}.

\subsection{The Bernoulli Displacement Model}

We conclude with a more detailed investigation of a one-dimensional special case of the displacement model, which exhibits some unexpected phenomena. Here we divide $\Z$ into neighboring pairs and, for each pair, randomly place a single site term into one of the two points of the pair. With our notation from above this corresponds to $d=1$, $M_1=2$, $b_1=1$, $q=\lambda \delta_1$, where $\lambda \in \mathbb{R}\setminus \{0\}$ is a fixed coupling constant, $\mathbb{P}(\omega_k=0)=p\in (0,1)$, and $\mathbb{P}(\omega_k=1)=1-p$. We will refer to this as the Bernoulli displacement model (BDM) and denote it by $h_{\omega,\lambda}$, keeping track of the dependence on the coupling constant.

\subsubsection{Almost Sure Spectrum}

We denote by $\Sigma_{\lambda}$ the almost-sure spectrum of $h_{\omega,\lambda}$. Theorem~\ref{thm: omega star} applies and thus the upper and lower edges $E_{\pm}(\lambda)$ of $\Sigma_{\lambda}$ are given by $\max \sigma(h_{\omega^*,\lambda})$ and $\min \sigma(h_{\omega^*,\lambda})$, respectively, where $\omega^*$ corresponds to the $4$-periodic potential with values $(0,\lambda,\lambda,0)$ in each period. By Floquet theory
\[ \sigma(h_{\omega^*,\lambda}) = \{E\in \R: \,|D(E)|\le 2\}, \]
where the discriminant $D(E)$, i.e.\ the trace of the monodromy operator, may be explicitly calculated as
\[ D(E)= E^4 - 2\lambda E^3 +(\lambda^2-4) E^2 +4\lambda E +2-\lambda^2.\]
The observation that $\sigma(h_{\omega^*,\lambda})$ is symmetric to $E=\lambda/2$ suggests to substitute $E=x+\lambda/2$, after which $D(\cdot)$ becomes bi-quadratic in $x$, which allows to explicitly determine the four bands of $\sigma(h_{\omega^*,\lambda})$. In particular, we find that
\[ E_{\pm}(\lambda) = \frac{\lambda}{2} \pm \sqrt{ 2+\frac{\lambda^2}{4} + \sqrt{4+\lambda^2}}. \]
Moreover, the central gap $(G_-(\lambda), G_+(\lambda))$ of $\sigma(h_{\omega^*,\lambda})$ is given by
\[ G_{\pm}(\lambda) = \frac{\lambda}{2} \pm \sqrt{2+\frac{\lambda^2}{4}-\sqrt{4+\lambda^2}}.\]

A deeper fact is that $(G_-(\lambda), G_+(\lambda))$ is a gap of $\sigma(h_{\omega,\lambda})$ for {\it every} configuration $\omega$, and thus also of the almost sure spectrum $\Sigma_{\lambda}$.

\begin{theorem}\label{thm: spectral gap}
For every $\lambda \in \mathbb{R}\setminus \{0\}$ and every $\omega\in \Omega$,
\begin{displaymath}
(G_-(\lambda),G_+(\lambda)) \cap \sigma(h_{\omega,\lambda}) = \emptyset.
\end{displaymath}
\end{theorem}

The $4$-periodic operator $h_{\omega^*,\lambda}$ has two additional non-trivial gaps, located symmetrically to the left and right of $(G_-(\lambda), G_+(\lambda))$. However, at least for $|\lambda|\le 2$, these gaps are filled in entirely by spectra from other configurations. In fact, consider $\sigma(h_{\omega^1,\lambda})$ for the constant configuration $\omega^1 := (\ldots, 1, 1, \ldots)$ giving the $2$-periodic potential $(\dots, 0, \lambda, 0, \lambda, \ldots)$. As it turns out, for details see \cite{Nichols}, for $|\lambda|\le 2$ the two bands of $\sigma(h_{\omega^1,\lambda})$ fully cover the left and right gaps of $\sigma(h_{\omega^*,\lambda})$. Thus, combining Theorems~\ref{thm: deterministic spectrum}, \ref{thm: spectral gap} and Corollary \ref{cor: upper and lower edges}, we get

\begin{corollary}\label{cor: lambda less than 2}
If $|\lambda|\leq 2$, then
\begin{displaymath}
\Sigma_{\lambda}=[E_-(\lambda),E_+(\lambda)]\setminus (G_-(\lambda), G_+(\lambda)).
\end{displaymath}
\end{corollary}

If $|\lambda|>2$, then the bands of $\sigma(h_{\omega^1,\lambda})$ cover the left and right gaps of $\sigma(h_{\omega^*,\lambda})$ only partially. In Section~\ref{sec:discussion} we state a conjecture on the structure of $\Sigma_{\lambda}$ for $|\lambda|>2$.

\subsubsection{Integrated Density of States}

For the one-dimensional {\it symmetric} BDM (i.e.\ the case $p=1/2$), the integrated density of states (IDS) shows surprising behavior near band edges. A similar result, meaning in particular that the IDS is not H\"older continuous at certain energies, was first shown in an analogous setting for the continuum displacement model at the bottom of the spectrum in \cite{BLS2}. For the discrete case considered here we get that the same phenomenon appears not only at the lower and upper edges of the almost sure spectrum, but also at the edges of the central gap identified above.

To define the IDS, let $L\in \mathbb{N}$. For $k\in \mathbb{Z}$, we set $\Lambda_k=[2k-1,2k]$ and define $\Lambda(L)=\cup_{i=1}^L \Lambda_i = [1,2L]$. Let $h_{\omega,\lambda}^L$ denote a restriction of $h_{\omega,\lambda}$ to $\Lambda(L)$ with arbitrary boundary condition (i.e.\ choice of the diagonal matrix elements at $-L$ and $L$). Set
\begin{displaymath}
E_1(h_{\omega,\lambda}^L)\leq E_2(h_{\omega,\lambda}^L)\leq \ldots,
\end{displaymath}
the eigenvalues of $h_{\omega,\lambda}^L$ counted with multiplicity.  The IDS of $h_{\omega,\lambda}$ at $E\in \mathbb{R}$ is
\begin{displaymath}
N_{\lambda}(E)=\lim_{L\rightarrow \infty}\frac{1}{|\Lambda(L)|}\mathbb{E}(\# \{k\in \mathbb{N}:E_k(h_{\omega,\lambda}^L)\leq E\}),
\end{displaymath}
which exists due to ergodicity of our model, e.g.\ \cite{Kirsch}, and is independent of the boundary condition.

First, we note the following symmetry property which simplifies matters.

\begin{theorem}\label{thm: IDS symmetry}
 If $E_{\pm}(\lambda)$ denote the upper and lower edges of the almost-sure spectrum $\Sigma_{\lambda}$ of the symmetric BDM (i.e. $p=1/2$) and $N_{\lambda}$ denotes the IDS, then
 \begin{equation}\label{eq: symmetry of IDS}
 N_{\lambda}(E_-(\lambda)+t)=1-N_{\lambda}(E_+(\lambda)-t)
 \end{equation}
 for every $t\in \mathbb{R}$.
\end{theorem}

We want to describe the asymptotics of the IDS near the four band edges $E_{\pm}(\lambda)$ and $G_{\pm}(\lambda)$. Due to symmetry we only need to consider the lower band edges $E_-(\lambda)$ and $G_+(\lambda)$.

\begin{theorem}\label{thm: DOS blowup}
Fix $\lambda>0$ and let $E_0=E_-(\lambda)$ or $E_0=G_+(\lambda)$.  Then there exist constants $C,\epsilon>0$ such that
\begin{displaymath}
N_{\lambda}(E)-N_{\lambda}(E_0)\geq\frac{C}{\log^2(E-E_0)} \quad \textrm{for all} \ E\in (E_0,E_0+\epsilon).
\end{displaymath}
  A corresponding result, for energies $E$ to the left of $E_0$,  holds at the upper edges $E_+(\lambda)$ and $G_-(\lambda)$.
\end{theorem}

For more discussion, including a conjecture on the asymptotics of the IDS at possible additional band edges for the case $|\lambda| >2$, see Section~\ref{sec:discussion}.


\section{Bubbles Tend to the Corners and Consequences} \label{sec:multiDresults}

Our first goal in this section is to prove Theorem~\ref{thm: bubbles}, i.e.\ that ``bubbles tend to the corners''. This will be done in Section~\ref{sec:bubblesproof} after Section~\ref{sec:NeuLap} will introduce the discrete Neumann Laplacian and list its relevant properties. The characterization of the spectral minimum of the displacement model, i.e.\ Theorem~\ref{thm: omega star} and Corollary~\ref{cor: upper and lower edges} are consequences of Theorem~\ref{thm: bubbles} and will be proven in Section~\ref{sec:consequences}.

\subsection{The Neumann Laplacian and Basic Properties} \label{sec:NeuLap}

Let $\Lambda\subset \mathbb{Z}^d$.  The \textit{truncation operator} $h_{0,\Lambda}$ on $\Lambda$ is the operator on $\ell^2(\Lambda)$ with matrix elements
\begin{displaymath}
h_{0,\Lambda}(i,j)=\left\{\begin{array}{ll}
-1, & i,j\in \Lambda, \ \|i-j\|_{1}=1\\
0, & \textrm{otherwise}
\end{array}.\right.
\end{displaymath}
The \textit{edge counting function} on $\Lambda$ is the function $n_{\Lambda}:\Lambda \rightarrow \mathbb{N}\cup\{0\}$ given by
\begin{displaymath}
n_{\Lambda}(i)=\# \{j\in \mathbb{Z}^d\setminus \Lambda: (i,j)\in \partial \Lambda\}.
\end{displaymath}
Here $\partial \Lambda$ is the boundary of $\Lambda$, i.e.
\begin{displaymath}
\partial \Lambda =\{(i,j)\in \mathbb{Z}^d \times \mathbb{Z}^d: i\in \Lambda, \ j\notin \Lambda, \ \|i-j\|_{1}=1 \ \textrm{or} \ i\notin \Lambda, \ j\in \Lambda, \ \|i-j\|_{1}=1\}.
\end{displaymath}

Associated to the edge counting function is the edge counting operator $N_{\Lambda}$ specified by the matrix elements
\begin{equation}\label{eq: edge counting operator}
N_{\Lambda}(i,j)=\left\{\begin{array}{ll}
n_{\Lambda}(i), & i=j, \ i\in \Lambda\\
0,& \textrm{otherwise}
\end{array}.\right.
\end{equation}

\begin{definition}\label{def: Neumann Laplacian}
If $\Lambda\subset \mathbb{Z}^d$, then the $d$-dimensional discrete Neumann Laplacian on $\Lambda$ is the operator $h_{0,\Lambda}^N$ defined on $\ell^2(\Lambda)$ and given by the operator sum
\begin{displaymath}
h_{0,\Lambda}^N=h_{0,\Lambda}-N_{\Lambda}.
\end{displaymath}
The matrix elements of $h_{0,\Lambda}^N$ are given by summing the corresponding matrix elements of $h_{0,\Lambda}$ and $N_{\Lambda}$.
\end{definition}
We now summarize some basic properties of the Neumann Laplacian. The
above definition and most of these properties (in fact, probably
all) can be found in various references, e.g.\ \cite{Simon,
KirschMuller,Kirsch}. Detailed proofs of all of them can be found in
\cite{Nichols}.

The first property is a ``reflection principle'' for the Neumann Laplacian, meaning that a solution $u$ to the equation $(h_{0,\Lambda}^N+q)u=Eu$ can, by reflection, be used to construct a solution on a larger set $\Lambda^{\prime}$.  More precisely, let $\Lambda=\prod_{i=1}^d[a_i,b_i]$ and $q,u: \Lambda \rightarrow \mathbb{R}$ with
  \begin{equation}\label{reflectionEVeq}
  (h_{0,\Lambda}^N+q)u=Eu
  \end{equation}
  for some $E\in \mathbb{R}$.  Let $k\in \{1,\ldots,d\}$ be fixed and $u_{k,ref}$ (respectively, $q_{k,ref}$) be the extension of $u$ (respectively, $q$) to the set
  \begin{displaymath}
  \Lambda^{\prime}=\prod_{i=1}^{k-1}[a_i,b_i]\times [a_k,2b_k-a_k+1]\times \prod_{i=k+1}^d[a_i,b_i]
  \end{displaymath}
  obtained by reflecting $u$ (respectively, $q$) about $b_k+\frac{1}{2}$ in the $k$th component.  It is not too hard to show that $(h_{0,\Lambda^{\prime}}+q_{k,ref})u_{k,ref}=Eu_{k,ref}$.  We summarize this result, along with three others, in

  \begin{proposition}[\textbf{Properties of the Neumann Laplacian}]\label{prop: properties} \ \ \ \ \ \ \ \
  \begin{itemize}
  \item[(i)] \textbf{(Reflection Property)} If $u$ satisfies (\ref{reflectionEVeq}), then
  \begin{displaymath}
  (h_{0,\Lambda^{\prime}}+q_{k,ref})u_{k,ref}=Eu_{k,ref}.
  \end{displaymath}
  \item[(ii)] \textbf{(Neumann Splitting Formula)} If $\Lambda_1\subset \Lambda\subset \mathbb{Z}^d$, then
  \begin{displaymath}
  h_{0,\Lambda}^N\geq h_{0,\Lambda_1}^N\oplus h_{0,\Lambda\setminus \Lambda_1}^N,
  \end{displaymath}
  in the quadratic form sense.
  \item[(iii)] \textbf{(Simplicity and Positivity of Ground State)} If $\Lambda\subset \mathbb{Z}^d$ is connected and $q:\Lambda \rightarrow \mathbb{R}$, then the ground state eigenvalue of $h_{0,\Lambda}^N+q$ is simple and the corresponding eigenfunction may be taken strictly positive.
  \item[(iv)] \textbf{(Ground State Energy of $h_{0,\Lambda}^N$)}  If $\Lambda\subset \mathbb{Z}^d$, then $\min \sigma(h_{0,\Lambda}^N)=-2d$.
  \end{itemize}
  \end{proposition}

  These properties will be used frequently in the proofs of our results.

  \subsection{Proof of Theorem \ref{thm: bubbles}} \label{sec:bubblesproof}

  Theorem \ref{thm: bubbles} is a statement about $E_0(a_1,\ldots,a_d)$ as a function of the $i$th component $a_i$ with all other components held fixed.  Without loss of generality, we may assume $i=1$.  Then since $a_2,\ldots,a_d$ are held fixed, to simplify notation, we write
  \begin{eqnarray}
  \widehat{q}_{a_1}&:=&q_{(a_1,\ldots,a_d)}\nonumber\\
  H(a_1)&:=&h_{0,\Lambda}^N+\widehat{q}_{a_1}.\nonumber
  \end{eqnarray}
   For $a_1\in [r_1,M_1-b_1]\cap \mathbb{Z}$, let $u_{a_1}$ denote the (up to a constant multiple) unique positive ground state corresponding to $E_0(a_1,a_2,\ldots,a_d)$ known to exist by Proposition \ref{prop: properties},
    \begin{equation}\label{eq: gs eq for a1}
    H(a_1)u_{a_1}=E_0(a_1,a_2,\ldots,a_d)u_{a_1}.
    \end{equation}
      Let $u_{a_1}^{\ast}$ and $\widehat{q}_{a_1}^{\ast}$ denote the extensions of $u_{a_1}$ and $\widehat{q}_{a_1}$ to
    \begin{displaymath}
    \Lambda^{\ast}:=[1,2M_1]\times [1,M_2]\times \cdots \times [1,M_d]
    \end{displaymath}
    which are symmetric about the axis $x_1=M_1+\frac{1}{2}$.   Let $u_{a_1,ext}$ (respectively, $\widehat{q}_{a_1,ext}$) denote the extension of $u_{a_1}^{\ast}$ (respectively, $\widehat{q}_{a_1}^{\ast}$) to the strip
    \begin{displaymath}
    \mathbb{Z}\times [1,M_2]\times \cdots [1,M_d]
    \end{displaymath}
    which is $2M_1$-periodic in the first component.

    That $u_{a_1}$ satisfies the ground state equation (\ref{eq: gs eq for a1})  implies
    \begin{equation}\label{eq: star}
    (h_{0,\Lambda^{\ast}}^N+\widehat{q}_{a_1}^{\ast}|_{\Lambda^{\ast}})u_{a_1}^{\ast}=E_0(a_1,\ldots,a_d)u_{a_1}^{\ast}.
    \end{equation}

    Now we turn to the proof of the theorem.  Our goal is to show
    \begin{displaymath}
    E_0(a_1+1,a_2,\ldots,a_d)<E_0(a_1,a_2,\ldots,a_d)
    \end{displaymath}
    for all $a_1\in [r_1,M_1-b_1-1]$.  Let $a_1\in [r_1,M_1-b_1-1]$, note that the former implies that $\widehat{q}_{a_1}(1,x_2,\ldots,x_d)=0$ for all $x_i\in [1,M_i]$ and $i\in \{2,\ldots,d\}$ which guarantees that
    \begin{equation}\label{eq: 0 at ends}
    0=\widehat{q}_{a_1}^{\ast}(1,x_2,\ldots,x_d)=\widehat{q}_{a_1}^{\ast}(2M_1,x_2,\ldots,x_d).
    \end{equation}
    If $R$ is the right shift in the first coordinate, $(R\psi)(n_1,n_2,\ldots,n_d)=\psi(n_1-1,n_2,\ldots,n_d)$, and
    \begin{displaymath}
    f_{a_1,R}:=Rf_{a_1,ext}|_{\Lambda^{\ast}},
    \end{displaymath}
    for $f\in \{u,\widehat{q}\}$, then
    \begin{eqnarray}
   \lefteqn{ \langle (h_{0,\Lambda^{\ast}}^N+\widehat{q}_{a_1,R})u_{a_1,R}, u_{a_1,R}\rangle - \langle (h_{0,\Lambda^{\ast}}^N+\widehat{q}_{a_1}^{\ast}|_{\Lambda^{\ast}})u_{a_1}^{\ast},u_{a_1}^{\ast}\rangle} \nonumber \\
    &=& -\sum_{(x_2,\ldots,x_d)\in \widehat{\Lambda}}(u_{a_1}^{\ast}(1,x_2,\ldots,x_d)-u_{a_1}^{\ast}(2,x_2,\ldots,x_d))^2\nonumber\\
    &=&-\sum_{(x_2,\ldots,x_d)\in \widehat{\Lambda}}(u_{a_1}(1,x_2,\ldots,x_d)-u_{a_1}(2,x_2,\ldots,x_d))^2,\label{eq: form difference}
    \end{eqnarray}
    where
    \begin{displaymath}
    \widehat{\Lambda}:=[1,M_2]\times \ldots \times [1,M_d].
    \end{displaymath}
    We encourage the reader to check the above calculation of (\ref{eq: form difference}) (at least in the simplest case $d=1$) using the definition of the Neumann Laplacian, reflection symmetry of $\widehat{q}_{a_1}^{\ast}|_{\Lambda^{\ast}}$ in the first coordinate, and (\ref{eq: 0 at ends}).

    Obviously, the quantity (\ref{eq: form difference}) is non-positive.  If (\ref{eq: form difference}) vanishes, then the restriction of $u_{a_1}$ to the ``slab'' $K:=\{1\} \times [1,M_2] \times \ldots \times [1,M_d]$ turns out to be the ground state of the Neumann Laplacian on $K$, for details see \cite{Nichols}. Thus, by Proposition~\ref{prop: properties}, $E_0(a_1,\ldots,a_d)=-2d$. In other words,  if $E_0(a_1,a_2,\ldots,a_d)\neq -2d$, then the quantity in (\ref{eq: form difference}) is strictly negative.

    The assumptions of Theorem \ref{thm: bubbles}, in either case (i) or (ii), are enough to guarantee that $E_0(a_1,a_2,\ldots,a_d)\neq-2d$.  In case (i), this can be seen by noting that the ground state eigenvalue of $h_{0,\Lambda}^N$ is $-2d$, and that a sign-definite perturbation strictly increases ($q>0$) or decreases ($q<0$) the ground state eigenvalue.  In case (ii), we make use of the following fact:

 If $\min \sigma(h_{0,\Lambda}^N+q_a)=-2$ for some $a\in [0,M_1-b_1]$, then $\min \sigma(h_{0,\Lambda}^N+q_a)=-2$ for all $a\in [0,M_1-b_1]$. To see this, note that $\min \sigma(h_{0,\Lambda}^N+q_a) =-2$ implies that the ground state is constant outside the support of $q_a$. Thus the ground state can be shifted together with the potential to produce the ground state for other values of $a$, leaving the ground state energy unaffected.

Thus, under the assumptions of Theorem \ref{thm: bubbles}, the quantity (\ref{eq: form difference}) is strictly negative, meaning that

  \begin{equation}\label{eq: form comparison}
  \langle (h_{0,\Lambda^{\ast}}^N+\widehat{q}_{a_1,R})u_{a_1,R},u_{a_1,R}\rangle<\langle (h_{0,\Lambda^{\ast}}^N+\widehat{q}_{a_1}^{\ast}|_{\Lambda^{\ast}})u_{a_1}^{\ast},u_{a_1}^{\ast}\rangle=E_0(a_1,a_2,\ldots,a_d)\|u_{a_1}^{\ast}\|^2,
  \end{equation}
  where the equality follows from (\ref{eq: star}).

    By definition, $\|u_{a_1}^{\ast}\|=\|u_{a_1,R}\|$; applying the variational principle and using the splitting formula from Proposition \ref{prop: properties} along with (\ref{eq: form comparison}) gives
    \begin{eqnarray}
    E_0(a_1,a_2,\ldots,a_d)&>&\inf \sigma(h_{0,\Lambda^{\ast}}^N+\widehat{q}_{a_1,R})\label{eq: strict}\\
    &\geq&\inf \sigma(H(a_1-1)\oplus H(a_1+1))\nonumber\\
    &=&\min \{E_0(a_1-1,a_2,\ldots,a_d),E_0(a_1+1,a_2,\ldots,a_d)\}.\label{eq: take min}
    \end{eqnarray}

    We prove the theorem by induction on $a_1$ using (\ref{eq: strict}) and (\ref{eq: take min}).  The first step is to show
    \begin{displaymath}
    E_0(r_1,a_2,\ldots,a_d)>E_0(r_1+1,a_2,\ldots,a_d).
    \end{displaymath}
    We make use of (\ref{eq: take min}) above with $a_1=r_1$.  If $M_1-b_1$ is even then by reflection symmetry of $E_0(\cdot)$ in the first coordinate on $\Delta$,
    \begin{displaymath}
    E_0(r_1-1,a_2,\ldots,a_d)=E_0(r_1+1,a_2,\ldots,a_d),
    \end{displaymath}
    thus the minimum in (\ref{eq: take min}) is certainly $E_0(r_1+1,a_2,\ldots,a_d)$.  If $M_1-b_1$ is odd, then
    \begin{displaymath}
    E_0(r_1-1,a_2,\ldots,a_d)=E_0(r_1,a_2,\ldots,a_d),
    \end{displaymath}
    and because we have the strict inequality (\ref{eq: strict}), the minimum in (\ref{eq: take min}) is $E_0(r_1+1,a_2,\ldots,a_d)$.

    For the induction step, suppose $a_1-1,a_1,a_1+1\in [r_1,M_1-b_1]$ and
    \begin{equation} \label{eq: induction hypothesis}
    E_0(a_1-1,a_2,\ldots,a_d) \ge E_0(a_1,a_2,\ldots,a_d).
    \end{equation}
    Using (\ref{eq: take min}), we have
    \begin{equation}\label{eq: induction step}
    E_0(a_1,a_2,\ldots,a_d)>\min\{E_0(a_1-1,a_2,\ldots,a_d),E_0(a_1+1,a_2,\ldots,a_d)\}.
    \end{equation}
    With the induction hypothesis (\ref{eq: induction hypothesis}), the minimum in (\ref{eq: induction step}) must be $E_0(a_1+1,a_2,\ldots,a_d)$.

\subsection{Proof of Theorem \ref{thm: omega star} and Corollary~\ref{cor: upper and lower edges}}
\label{sec:consequences}

In addition to Theorem \ref{thm: bubbles}, the proof of Theorem \ref{thm: omega star} relies on the following two facts, which are discrete versions of results known as Allegretto-Piepenbrink Theorem and Shnol's Theorem in the continuum, e.g.\ \cite{SSG}.

\begin{theorem}[Allegretto-Piepenbrink]\label{thm: Allegretto-Piepenbrink}
Let $E\in \mathbb{R}$ and $V:\mathbb{Z}^d\rightarrow \mathbb{R}$ be bounded.  If there exists a non-trivial $u\geq 0$ satisfying the difference equation $(h_0+V)u=Eu$, then $E\leq \inf \sigma(h_0+V)$.
\end{theorem}

\begin{theorem}[Schnol]\label{thm: Schnol}
If $u$ is a polynomially bounded generalized eigenfunction of $h_0+V$ corresponding to $E\in \mathbb{R}$, then $E\in \sigma(h_0+V)$.
\end{theorem}

Proofs of these facts, most likely also well known, are provided in \cite{Nichols}.

\begin{proof}[Proof of Theorem \ref{thm: omega star}] (a)  Note that $h_{\omega}$ restricted to $\Lambda+kM$ with Neumann boundary condition is unitarily equivalent (via translation) to $h_{0,\Lambda}^N+q_{\omega_k}$.  Thus, the splitting formula for the Neumann Laplacian from Proposition \ref{prop: properties} gives
      \begin{eqnarray}
        \inf \ \sigma(h_{\omega})&\geq&\inf \ \sigma\bigg(\bigoplus_{k\in \mathbb{Z}^d}(h_{\Lambda}^N+q_{\omega_k}) \bigg)\nonumber\\
        &\geq&\inf_{k\in \mathbb{Z}^d} \ \min \ \sigma(h_{\Lambda}^N+q_{\omega_k})\nonumber\\
        &\geq&\min \ \{E_0(a):a\in \Delta\}\nonumber\\
        &=&E_0(a_{min}),\label{eq: lower bound on omega}
      \end{eqnarray}
      with the corner position $a_{min}=(M_i-b_i)_{i=1}^d.$   Thus, $E_0(a_{min})$ is a lower bound on the infimum of the spectrum of any operator $h_{\omega}$ with $\omega \in \Omega$.

      In view of (\ref{eq: lower bound on omega}), $\omega^{\ast}$ is spectrally minimizing if $\inf \ \sigma(h_{\omega^{\ast}})=E_0(a_{min})$. Applying the Allegretto-Piepenbrink and Schnol Theorems,  it is enough to show the existence of a strictly positive bounded function, $\psi$, which satisfies $h_{\omega^{\ast}}\psi=E_0(a_{min})\psi$ in the sense of a solution of a finite difference equation.

Let $u$ denote the strictly positive ground state of $h_{0,\Lambda}^N+q_{a_{min}}$, so that
      \begin{equation}\label{eq: little u}
      (h_{0,\Lambda}^N+q_{a_{min}})u=E_0(a_{min})u.
      \end{equation}
      Let $U$ and $Q$ denote the extensions of $u$ and $q_{a_{min}}$ to
      \begin{displaymath}
      \Lambda_2=\prod_{i=1}^d[1,2M_i],
      \end{displaymath}
      which are reflection symmetric on $\Lambda_2$ in each of the $d$ coordinates.   Looking at (\ref{eq: little u}), we immediately have
      \begin{equation} \label{eq: big u}
      (h_{0,\Lambda_2}^N+Q)U=E_0(a_{min})U.
      \end{equation}
      If $\psi$ is the extension of $U$ to all of $\mathbb{Z}^d$ which is $2M_i$-periodic in the $i$th coordinate, then reflection symmetry of $U$ on $\Lambda_2$ and (\ref{eq: big u}) give
      \begin{displaymath}
      h_{\omega^{\ast}}\psi=E_0(a_{min})\psi.
      \end{displaymath}

      (b) To show $\omega^{\ast}$ is spectrally maximizing for the displacement model defined with single-site potential $q$, it suffices to show $\omega^{\ast}$ is spectrally minimizing for the model with single-site potential $-q$.  In fact, applying spectral mapping and the unitary involution
      \begin{equation}\label{eq: unitary involution}
      (U\phi)(n)=(-1)^{\sum_{i=1}^d |n_i|}\phi(n), \qquad n\in \mathbb{Z}^d, \ \phi\in \ell^2(\mathbb{Z}^d),
      \end{equation}
      we have $Uh_0U=-h_0$, and thus, for any $\omega \in \Omega$,
      \begin{eqnarray} \label{eq:unitaryequiv}
      \max \sigma(h_0+V_{\omega})&=&-\min \sigma(-(h_0+V_{\omega}))\nonumber\\
      &=&-\min \sigma(U(-h_0-V_{\omega})U)\nonumber\\
      &=&-\min \sigma(h_0-V_{\omega}).\label{eq: min implies max}
      \end{eqnarray}
      If $\omega^{\ast}$ is spectrally minimizing for $-q$, then (\ref{eq: min implies max}) gives
      \begin{eqnarray}
      \max \sigma(h_0+V_{\omega})&\leq&-\min \sigma(h_0-V_{\omega^{\ast}})\nonumber\\
      &=&\max \sigma(-h_0+V_{\omega^{\ast}})\nonumber\\
      &=&\max \sigma(h_0+V_{\omega^{\ast}}),\nonumber
      \end{eqnarray}
      for any $\omega \in \Omega$.  Therefore, $\omega^{\ast}$ is spectrally maximizing for $q$.

      That $\omega^{\ast}$ is spectrally minimizing for $-q$ follows immediately from (a):  If $d\geq2$, sign-definiteness of $q$ implies sign-definiteness of $-q$.  By (a), sign-definiteness of $-q$ guarantees that $\omega^{\ast}$ is spectrally minimizing for $-q$.  If $d=1$, the assumption $\min \sigma(h_{0,\Lambda}^N-q_{\widetilde{a}})\neq -2$ together with (a) implies $\omega^{\ast}$ is spectrally minimizing for $-q$.

\end{proof}

\begin{proof}[Proof of Corollary \ref{cor: upper and lower edges}]
Under the assumptions of Theorem~\ref{thm: omega star}(a) we have
\[\min \sigma(h_{\omega^{\ast}}) \leq \min \sigma(h_{\omega})\] for every $\omega \in \Omega$.  Thus $\min \sigma(H_{\omega^*}) \le \inf \Sigma$.
Since $\omega^{\ast}\in \mathcal{C}_{per}$, we have by Theorem~\ref{thm: deterministic spectrum} that $\sigma(h_{\omega^*}) \subset \Sigma$ and, in particular, $\min \sigma(h_{\omega^*}) \ge \inf \Sigma$. We conclude $\inf \Sigma = \min \sigma(h_{\omega^*})$.

Similar, under the assumptions of Theorem~\ref{thm: omega star}(b), it follows that \[\sup \Sigma = \max \sigma(h_{\omega^*}).\]
\end{proof}

\section{Non-uniqueness of the One-dimensional Minimizer} \label{sec:nonunique}

\subsection{The Discrete Periodic Laplacian}
  Suppose $S \subset \mathbb{Z}^d$ has the form
  \begin{equation}\label{lambdaform}
    S=\prod_{i=1}^d [a_i,b_i]
  \end{equation}
  for $a_i,b_i \in \mathbb{Z}$, with $a_i\leq b_i$ for all $1\leq i \leq d$.  For $u:S\rightarrow \mathbb{C}$, define the function $u_T:\mathbb{Z}^d\rightarrow \mathbb{C}$ by
  \begin{displaymath}
    u_T(n)=\sum_{k\in \mathbb{Z}^d}u(n-kL).
  \end{displaymath}
  Here $L=(L_i)_{i=1}^d$ with $L_i:=b_i-a_i+1$ and $kL:=(k_iL_i)_{i=1}^d$.  Thus, $u_T$ is the $L$-periodic extension of $u$ to all of $\mathbb{Z}^d$. One may view $u_T$ as $u$ defined on a $d$-dimensional torus.

    The discrete Laplacian on $S$ with periodic boundary condition, or the periodic Laplacian, denoted $h_{0,S}^P$, has domain $\ell^2(S)$ and acts on a function $u\in \ell^2(S)$ by
    \begin{displaymath}
      (h_{0,S}^Pu)(n)=(h_0u_T)(n)\qquad \mbox{for all}\; n\in S.
    \end{displaymath}

  Thus, if $u_T$ is viewed as $u$ on a $d$-dimensional torus, then $h_{0,S}^P$ is simply the discrete Laplacian on the torus.  We summarize some properties of the periodic Laplacian for which proofs can be found in \cite{Nichols}:

  \begin{proposition}[Properties of the Periodic Laplacian]\label{prop: properties of per}\ \ \ \ \ \ \ \ \ \
  \begin{itemize}
  \item[(i)] If $V:S\rightarrow \mathbb{R}$, then the ground state of $h_{0,S}^P+V$ is simple and the corresponding eigenspace is spanned by a strictly positive eigenfunction.
  \item[(ii)] $h_{0,S}^P\geq h_{0,S}^N$ in the form sense.
  \item[(iii)] If $V:\mathbb{Z}^d\rightarrow \mathbb{R}$ is $L$-periodic, then $\inf \sigma(h_0+V)=\inf \sigma(h_{0,C}^P+V|_C)$, where $C$ is any period cell of $V$.
  \item[(iv)] Let $u$ and $V$ be functions from $[a,b]$ (as an interval in $\mathbb{Z}$) into $\mathbb{R}$ and $E\in \mathbb{R}$.  If $u(a)=u(b)$, then $(h_{0,[a,b]}^N+V)u=Eu$ if and only if $(h_{0,[a,b]}^P+V)u=Eu$.  In particular, if $V$ is reflection symmetric, then the ground states of $h_{0,[a,b]}^N$ and $h_{0,[a,b]}^P$ coincide.
  \end{itemize}
  \end{proposition}

  \subsection{A Preliminary Result}

  We begin with a result that provides a necessary criterion for a periodic configuration $\omega$ to be spectrally minimizing for $h_{\omega}$. This result holds for arbitrary dimension $d\ge 1$.

   Suppose $\omega \in \Omega$ is an $L$-periodic configuration so that
  \begin{displaymath}
    \omega_{i+(n_1L_1,n_2L_2,\ldots,n_dL_d)}=\omega_{i+nL}=\omega_i
  \end{displaymath}
  for every $i\in \mathbb{Z}^d$ and every $n=(n_1,n_2,\ldots,n_d)\in \mathbb{Z}^d$.  It follows that the potential $V_{\omega}$ given by (\ref{eqn: random potential}) is $\widehat{L}$-periodic, where $\widehat{L}=LM=(L_iM_i)_{i=1}^d$ and a period cell for $V_{\omega}$ is
  \begin{equation}\label{eqn: period cell for V}
    \underline{\Lambda}=\prod_{i=1}^d[1,L_iM_i].
  \end{equation}
   The cell $\underline{\Lambda}$ can be written as a disjoint union of translates of the basic cell $\Lambda = \prod_{i=1}^d [1,M_i]$, by setting $\mathcal{K}=\prod_{i=1}^d[1,L_i]$
  and defining $\Lambda_k=\Lambda+(k-1)M$ for $k\in \mathcal{K}$.  Here $k-1$ is the vector with $i$th component $k_i-1$.

  Before proceeding with our preliminary result, we fix the following notations.  For $\bullet \in \{N,P\}$, let $h_{\omega,S}^{\bullet}$ denote the Neumann ($\bullet=N$) or periodic ($\bullet=P$) restriction of $h_{\omega}$ to the set $S\subset \mathbb{Z}^d$ and let $E_0(h_{\omega,S}^{\bullet})$ denote the corresponding ground state eigenvalue.  For $\omega_i\in \Delta$, let $E_0(\omega_i)$ denote the ground state eigenvalue of $h_{0,\Lambda}^N+q_{\omega_i}$.  Lastly, set as before $E_{min}=\inf_{\omega\in \Omega} \min \sigma(h_{\omega})$.

  The following result shows that in a spectrally minimizing periodic configuration $\omega$, every $\omega_i$ is a corner of $\Delta$, i.e.\ it lies in
  \[ {\mathcal C} :=  \{(a_1,\ldots,a_d)\in\Delta: a_i\in\{0,M_i-b_i\} \ \mbox{for all  $i=1,\ldots,d$}\}.\]

    \begin{lemma}\label{lemma: minimizers sit in corners}
    If $\omega\in \Omega$ is a periodic configuration with $\min \sigma(h_{\omega})=E_{min}$, then $\omega_k\in {\mathcal C}$ for every $k\in \mathbb{Z}^d$.  Moreover, in this case, $E_{min}=E_0(h_{\omega,\underline{\Lambda}}^N)=E_0(h_{\omega,\underline{\Lambda}}^P)$, where $\underline{\Lambda}$ is the period cell for $V_{\omega}$ given by (\ref{eqn: period cell for V}).  If $\psi_{\omega}$ is the positive ground state for $h_{\omega,\underline{\Lambda}}^N$ with ground state eigenvalue $E_{min}$, then, for every $k\in {\mathcal K}$, $\psi_{\omega}|_{\Lambda_k}$ is (up to normalization) the positive ground state for $h_{\omega,\Lambda_k}^N$ with ground state eigenvalue $E_0(\omega_k)=E_{min}$.
  \end{lemma}

  \begin{proof}
    By part (iii) of Proposition \ref{prop: properties of per}, $E_{min}=\min \sigma(h_\omega)=\min \sigma(h_{\omega,\underline{\Lambda}}^P)$. Part (ii) of the same proposition along with the variational principle gives $E_0(h_{\omega,\underline{\Lambda}}^P)\geq E_0(h_{\omega,\underline{\Lambda}}^N)$.  Moreover, since $\psi_{\omega}$ minimizes the the quadratic form for $h_{\omega,\underline{\Lambda}}^N$,
    \begin{eqnarray}
      E_0(h_{\omega,\underline{\Lambda}}^N)&=&\frac{(\psi_{\omega},h_{\omega,\underline{\Lambda}}^N\psi_{\omega})}{(\psi_{\omega},\psi_{\omega})}\nonumber\\
      &\geq&\frac{(\psi_{\omega},(\bigoplus_{k\in\mathcal{K}}h_{\omega,\Lambda_k}^N)\psi_{\omega})}{(\psi_{\omega},\psi_{\omega})}\nonumber\\
      &=&\frac{\sum_{k\in\mathcal{K}}(\psi_{\omega}^k,h_{\omega,\Lambda_k}^N\psi_{\omega}^k)}{(\psi_{\omega},\psi_{\omega})}\nonumber\\
      &=&\sum_{k\in \mathcal{K}}\frac{(\psi_{\omega}^k,h_{\omega,\Lambda_k}^N\psi_{\omega}^k)}{(\psi_{\omega}^k,\psi_{\omega}^k)}\cdot\frac{(\psi_{\omega}^k,\psi_{\omega}^k)}{(\psi_{\omega},\psi_{\omega})}\nonumber\\
      &\geq&\sum_{k\in \mathcal{K}}E_0(\omega_k)\frac{(\psi_{\omega}^k,\psi_{\omega}^k)}{(\psi_{\omega},\psi_{\omega})}\label{firstinstring}\\
      &\geq&\sum_{k\in \mathcal{K}}E_{min}\frac{(\psi_{\omega}^k,\psi_{\omega}^k)}{(\psi_{\omega},\psi_{\omega})} \;\; = \;\; E_{min}  \label{secondinstring}
    \end{eqnarray}
    where $\psi_{\omega}^k:=\psi_{\omega}|_{\Lambda_k}$.  The second to last inequality follows from the variational principle and the fact that $E_0(\omega_k)$ is by definition the lowest eigenvalue of $h_{\omega,\Lambda_k}^N$.  The last inequality follows from the fact that $E_0(\omega_k)\geq E_{min}$ as $E_0(\omega_k)\geq E_0(c)=E_{min}$ by Theorem~\ref{thm: bubbles} for any $c\in {\mathcal C}$.  Division by the quantity $(\psi_{\omega}^k,\psi_{\omega}^k)$ is permitted since strict positivity of $\psi_{\omega}$ implies strict positivity of $\psi_{\omega}^k$ and thus, $(\psi_{\omega}^k,\psi_{\omega}^k) >0$.

    We conclude that for a periodic minimizing configuration $\omega$ all inequalities in the above chain are actually identities.

    If at least one $\omega_k$ does not belong to ${\mathcal C}$, then the last inequality, (\ref{secondinstring}), is strict in light of Theorem \ref{thm: bubbles}.  Thus, $\omega_k\in {\mathcal C}$ for all $k\in \mathcal{K}$.  The fact that
    \begin{equation}\label{littleinequality}
      \frac{(\psi_{\omega}^k,h_{\omega,\Lambda_k}^N\psi_{\omega}^k)}{(\psi_{\omega}^k,\psi_{\omega}^k)}\geq E_0(\omega_k) \qquad \mbox{for all $k\in \mathcal{K}$}
    \end{equation}
    was used to obtain (\ref{firstinstring}), and if at least one of the inequalities (\ref{littleinequality}) is strict, then (\ref{firstinstring}) is strict.  Thus, there must be equality in (\ref{littleinequality}) for every $k\in \mathcal{K}$.  Therefore, $\psi_{\omega}^k$ is the ground state for $h_{\omega,\Lambda_k}^N$ and $E_0(\omega_k) = E_{min}$ is the corresponding lowest eigenvalue.
  \end{proof}

     \subsection{Characterization of One-dimensional Spectrally Minimizing Configurations}

     Now, we specialize to the one-dimensional setting.  The proof of Theorem \ref{thm: all 1d minimizers} relies on the following basic fact about the ``shape'' of a one-dimensional ground state eigenfunction.

     \begin{lemma}\label{lemma: unequal at endpoints}
Suppose $a\in [0,M_1-b_1]\cap \mathbb{Z}$ and $a\neq \frac{M-b}{2}$.  If $\psi$ is the strictly positive ground state corresponding to $h_{0,[1,M_1]}^N+q_a$ with ground state eigenvalue $E$ and $E\neq -2$, then  $\psi(1)\neq \psi(M)$.
\end{lemma}

\begin{proof}
By symmetry of the potential $q$ in $[1,b_1]$ and uniqueness of the ground state, it is enough to consider $a\in [r,M_1-b_1]$, where $r$ is the least integer strictly greater than $(M_1-b_1)/2$.  We may take $\psi$ to be normalized.

Let $q_{a,ref}$ denote the reflection of $q_a$ at the center of $[1,M_1]$, i.e.
\begin{displaymath}
q_{a,ref}(n)=q_a(M_1+1-n), \qquad n\in [1,M_1].
\end{displaymath}
By symmetry of $q$ on $[1,b_1]$, the normalized strictly positive ground state eigenfunction of $h_{0,[1,M_1]}^N+q_{a,ref}$ is $\psi_{ref}$, the reflection of $\psi$ on $[1,M_1]$, and the corresponding ground state eigenvalue is $E$.

If $\psi(1)=\psi(M_1)$, then it follows from part (iv) of Proposition \ref{prop: properties of per} that the ground state eigenvalue of $h_{0,[1,M_1]}^P+q_a$ is $E$ and $\psi$ is the corresponding eigenfunction:
\begin{equation}\label{eq: per restriction}
(h_{0,[1,M_1]}^P+q_a)\psi=E\psi.
\end{equation}
Let $Q$ and $\Psi$ denote the $M_1$-periodic extensions of $q_a$ and $\psi$, respectively, to all of $\mathbb{Z}$.  Using the definition of the periodic Laplacian and (\ref{eq: per restriction}) gives
\begin{displaymath}
(h_0+Q)\Psi=E\Psi
\end{displaymath}
which implies by Theorems~\ref{thm: Allegretto-Piepenbrink} and \ref{thm: Schnol} that $E_{min}=\inf \sigma(h_0+Q)$, since $\Psi$ is strictly positive and bounded (thus, polynomially bounded).  Therefore, by part (iii) of Proposition \ref{prop: properties of per}, $E$ is the lowest eigenvalue of the periodic restriction of $h_0+Q$ to any period cell, and $\Psi$ restricted to the same period cell is the corresponding ground state eigenfunction.

The trick is to choose a period cell $C$ for which $Q|_{C}$ is, up to translation, identical to $q_{a,ref}$.  The appropriate period cell is $$C=[b_1+2a+1-M_1,b_1+2a].$$ Note that the assumption $a\geq r$ implies that $M_1+1\in C$.  We have
\begin{eqnarray}
(h_{0,C}^P+Q|_{C})\Psi|_{C}&=&E \Psi|_{C}\nonumber\\
\Psi|_{C}(M_1)=\psi(M_1)&=&\psi(1)=\Psi|_{C}(M_1+1).
\end{eqnarray}
Since $C$ is a period cell for $\Psi$, we also have $\|\Psi|_{C}\|=\|\psi\|=\|\psi_{ref}\|$.

By construction, $\Psi|_{C}$ must be a constant multiple of the translate of $\psi_{ref}$ to $C$.  Normalization requires this constant multiple to be one, thus
\begin{displaymath}
\Psi|_{C}(n)=\psi_{ref}(n-(b_1+2a-M_1)), \qquad n\in C.
\end{displaymath}
Therefore, we have that
\begin{eqnarray}
\Psi|_{C}(M_1)&=&\psi_{ref}(2M_1-b_1-2a)\nonumber\\
\Psi|_{C}(M_1+1)&=&\psi_{ref}(2M_1-b_1-2a+1),
\end{eqnarray}
and it follows that
\begin{equation}\label{eq: contradiction}
\psi(b_1+2a-M_1+1)=\psi(b_1+2a-M_1).
\end{equation}
Note that $b_1+2a-M_1<1+a$ and $1+a$ is the minimum of the support of $q_a$.  Therefore, (\ref{eq: contradiction}) is a contradiction to the basic fact that the ground state $\psi$ is either strictly increasing (if $E<-2$) or strictly decreasing (if $E>-2$) below the support of $q_a$, \cite[Lemma 3.9]{Nichols}.

\end{proof}

Now we have everything we need for our proof of Theorem \ref{thm: all 1d minimizers}.

\begin{proof}[Proof of Theorem \ref{thm: all 1d minimizers}]
It follows from (\ref{eq:unitaryequiv}) that $\omega$ is spectrally maximizing for $h_0+V_{\omega}$ if and only if it is spectrally minimizing for $h_0-V_{\omega}$. Thus the claim made in Theorem~\ref{thm: all 1d minimizers} for spectral maximizers follows from the result on minimizers. It therefore remains to prove the latter.

   Let $\omega$ be an $L$-periodic minimizing configuration, so that $E_{min}=\min \sigma(h_{\omega})$.  By Lemma \ref{lemma: minimizers sit in corners}, it must be that $\omega_i\in \{0,M_1-b_1\}$ for all $i\in \mathbb{Z}$ and that
    \begin{equation}\label{eqofgs}
      E_{min}=\min \sigma(h_{\omega,\underline{\Lambda}}^P)=\min \sigma(h_{\omega,\underline{\Lambda}}^N),
    \end{equation}
    where $\underline{\Lambda}$ is the period cell $[1,LM_1]$.  Therefore, $\omega\in S_L$ and $n^0(\omega)+n^1(\omega)=L$.

    Let $\psi$ denote the strictly positive ground state of the periodic operator $h_{\omega,\underline{\Lambda}}^P$ corresponding to $E_{min}$.  Applying part (ii) of Proposition \ref{prop: properties of per}, along with the variational principle and (\ref{eqofgs}), gives
    \begin{displaymath}
      E_{min}=\frac{(\psi,h_{\omega,\underline{\Lambda}}^P\psi)}{(\psi,\psi)}\geq \frac{(\psi,h_{\omega,\underline{\Lambda}}^N\psi)}{(\psi,\psi)}\geq \inf_{\phi\neq 0}\frac{(\phi,h_{\omega,\underline{\Lambda}}^N\phi)}{(\phi,\phi)}=E_{min}.
    \end{displaymath}
    Thus, both inequalities above are actually equalities and $\psi$ is the ground state of $h_{\omega, \underline{\Lambda}}^N$ corresponding to $E_{min}$.  From the ground state equations $h_{\omega,\underline{\Lambda}}^{\bullet}\psi=E_{min}\psi$ for $\bullet \in \{N,P\}$, one gets
    \begin{eqnarray}
      -\psi(1)-\psi(2)+V_{\omega}(1)\psi(1)&=&E_{min}\psi(1), \nonumber\\
      -\psi(LM_1)-\psi(2)+V_{\omega}(1)\psi(1)&=&E_{min}\psi(1).\nonumber
    \end{eqnarray}
    The two together imply
    \begin{equation}\label{eqofendpoints}
      \psi(1)=\psi(LM_1).
    \end{equation}

   Set $\Lambda_k=\Lambda+(k-1)M_1$ with $k\in \{0,\ldots,L\}$.  Let $\phi_1$ denote the positive ground state of $h_{0,\Lambda}^N+q_{M_1-b_1}$ normalized so that $\phi_1(1)=1$, and let $\phi_{-1}$ denote the positive ground state of $h_{0,\Lambda}^N+q$ normalized so that $\phi_{-1}(1)=1$.

   Since the potential $q$ is reflection symmetric, $\phi$ given by $\phi(1+n)=\phi_1(M_1-n)$ for $n\in [0,M_1-1]$ is also a positive ground state for $h_{0,\Lambda}^N+q$.  Uniqueness of the positive ground state provides that $\phi=\phi_1(M_1)\phi_{-1}$.  Thus,
   \begin{displaymath}
   1=\phi_1(1)=\phi(M_1)=\phi_1(M_1)\phi_{-1}(M_1).
   \end{displaymath}
   Therefore, if $\alpha_i(\omega)= 1$ (respectively, $\alpha_i(\omega)=-1$) when $\omega_i= M_1-b_1$ (respectively, $\omega_i=0$), then $\phi_{\alpha_i(\omega)}(M_1)=\phi_1(M_1)^{\alpha_i(\omega)}$.  Lemma \ref{lemma: unequal at endpoints} guarantees that $\phi_1(M_1)\neq1$.  For simplicity of notation, set $\rho:=\phi_1(M_1)$.

   Using the reflection property of the Neumann Laplacian, the function $\Psi_{\omega}$, constructed by concatenating rescaled copies of $\phi_{\pm 1}$, and defined piecewise on each $\Lambda_{k}$, $k\in \{1,\ldots, L\}$ by

   \begin{displaymath}
   \Psi_{\omega}(n)=\rho^{\sum_{i<k}\alpha_i(\omega)}\phi_{\alpha_k(\omega)}(n-kM_1), \qquad n\in \Lambda_{k}
   \end{displaymath}
   (with the convention that an empty sum is zero) satisfies $h_{\omega,\underline{\Lambda}}^N \Psi_{\omega}=E_{min}\Psi_{\omega}$.  Therefore, we have recovered, up to a constant multiple, the ground state $\psi$.  Moreover,
   \begin{displaymath}
   \Psi_{\omega}(kM_1)=\rho^{\sum_{i=1}^k \alpha_i(\omega)}.
   \end{displaymath}
   This is seen as follows:
   \begin{displaymath}
   \Psi_{\omega}(kM_1)=\rho^{\sum_{i<k}\alpha_i(\omega)}\phi_{\alpha_k(\omega)}(M_1)=\rho^{\sum_{i<k}\alpha_i(\omega)}\rho^{\alpha_k(\omega)}=\rho^{\sum_{i=1}^{k}\alpha_i(\omega)}.
   \end{displaymath}
   Since $\psi$ coincides at the endpoints of $[1,LM_1]$, (\ref{eqofendpoints}), it must be that
   \begin{displaymath}
   1=\Psi_{\omega}(1)=\Psi_{\omega}(LM_1)=\rho^{\sum_{i=1}^{L}\alpha_i(\omega)}=\rho^{n^0(\omega)-n^1(\omega)}.
   \end{displaymath}
   Since $\rho\neq 1$, we conclude $n^0(\omega)=n^1(\omega)$.  Using $n^1(\omega)+n^0(\omega)=L$ gives $2n^1(\omega)=L$, so $L$ is even and $n^0(\omega)=n^1(\omega)=L/2$.
\end{proof}

\section{The Bernoulli Displacement Model} \label{sec:Bernoulli}

In this section we carry out the remaining proofs of Theorems~\ref{thm: spectral gap}, \ref{thm: IDS symmetry} and \ref{thm: DOS blowup}, which establish properties of the Bernoulli displacement model.

\subsection{Almost Sure Spectrum} \label{sec:asspectrum}

\begin{proof}[Proof of Theorem \ref{thm: spectral gap}]
By spectral mapping, it suffices to show that
\begin{displaymath}
H_{\omega,\lambda}:=\bigg(h_{\omega,\lambda}-\frac{\lambda}{2} \bigg)^2-\bigg(2+ \frac{\lambda^2}{4} \bigg)\geq -\sqrt{4+\lambda^2}\qquad \mbox{for all  $\omega \in \Omega$}.
\end{displaymath}
Fix $\omega \in \Omega$.  The operator $H_{\omega,\lambda}$ is a five diagonal operator
  \begin{displaymath}
  H_{\omega,\lambda}=
    \left(\begin{array}{ccccccccc}
    \ddots&\ddots&\ddots&\ddots&\ddots& & &\\
    &1&s_{\omega}(-2)&0&s_{\omega}(-1)&1&&\\
    &&1&s_{\omega}(-1)&0&s_{\omega}(0)&1&\\
    &&&1&s_{\omega}(0)&0&s_{\omega}(1)&1\\
    &&&&\ddots&\ddots&\ddots&\ddots&\ddots
    \end{array}\right),
  \end{displaymath}
  where
  \begin{displaymath}
  s_{\omega}(n)=-\lambda - V_{\omega}(n) - V_{\omega}(n+1).
  \end{displaymath}
  Since $V_{\omega}$ assumes both $0$ and $\lambda$ on each cell $[2k-1,2k]$, $k\in \mathbb{Z}$, it must be the case that $s_{\omega}(2k-1)=0$, i.e.\ $s_{\omega}$ vanishes at all odd lattice sites. On the other hand, at even sites an inspection of cases shows that
  \[ s_{\omega}(2k) = \lambda (\omega_{k+1} - \omega_k).\]
  Consequently, $H_{\omega,\lambda}$ is an infinite tridiagonal block matrix $H_{\omega,\lambda}=[B_{jk}]_{j,k \in \mathbb{Z}}$ where each block $B_{jk}$ is a $2\times 2$ matrix with $B_{jk}=\mathbb{I}$ if $|j-k|=1$ and $-1\leq j,k \leq 1$ and
  \begin{displaymath}
  B_{kk}=\left(\begin{array}{cc}
  0 & s_{\omega}(2k)\\
  s_{\omega}(2k) & 0
  \end{array}\right).
  \end{displaymath}

  Because of the $2\times 2$ block structure, it is instructive to view $H_{\omega,\lambda}$ as a Jacobi-type operator on $\ell^2(\mathbb{Z},\mathbb{C}^2)$.  That is, with $u\in \ell^2(\mathbb{Z},\mathbb{C}^2)$, $u(k)=(u_1(k),u_2(k))^{\intercal}$, we write
  \begin{displaymath}
  (H_{\omega,\lambda}u)(k)=u(k-1)+u(k+1)+B_{kk}u(k).
  \end{displaymath}
  Each $B_{kk}$ may be diagonalized via the transformation
  \begin{displaymath}
      W=\frac{1}{\sqrt{2}}\left(\begin{array}{cc}
      1&1\\
      1&-1
    \end{array}\right)
  \end{displaymath}
  which induces a unitary involution $U$ on $\ell^2(\mathbb{Z},\mathbb{C}^2)$ given by
  \begin{displaymath}
  (Uu)(k)=Wu(k) \qquad \mbox{for all  $k\in \mathbb{Z}$}.
  \end{displaymath}
  $H_{\omega,\lambda}$ is unitarily equivalent, via $U$, to another infinite tridiagonal block matrix $M_{\omega,\lambda}=[M_{jk}]_{j,k \in \mathbb{Z}}$ with each $M_{jk}$ a $2 \times 2$ matrix;  $M_{jk}=\mathbb{I}$ if $|j-k|=1$ and $-1\leq j, k \leq 1$ and $M_{kk}=\textrm{diag}(s_{\omega}(2k),-s_{\omega}(2k))$.

 The fact that $\mathbb{I}$ and $M_{kk}$ are diagonal means that $M_{\omega,\lambda}$ decouples into a direct sum of ``even'' and ``odd'' parts.  That is, if $u\in \ell^2(\mathbb{Z})$ is expressed as the direct sum of its even and odd parts, corresponding to even and odd components, $u=u_{\textrm{odd}}\oplus u_{\textrm{even}}$, then
  \begin{displaymath}
  M_{\omega,\lambda}u=M_{\omega,\lambda}^+u_{\textrm{odd}}\oplus M_{\omega,\lambda}^-u_{\textrm{even}},
  \end{displaymath}
  where $M_{\omega,\lambda}^{\pm}$ are both discrete Schr\"{o}dinger operators
  \begin{displaymath}
  M_{\omega,\lambda}^{\pm}=-h_0\pm q_{\omega}\cong h_0\pm q_{\omega},
  \end{displaymath}
  and the potential term $q_{\omega}$ is defined in terms of $s_{\omega}$ by
  \begin{displaymath}
  q_{\omega}(k)=s_{\omega}(2k) = \lambda (\omega_{k+1}-\omega_k).
  \end{displaymath}

  In light of the fact that
  \begin{displaymath}
  M_{\omega,\lambda}=M_{\omega,\lambda}^+\oplus M_{\omega,\lambda}^-,
  \end{displaymath}
  the proof of Theorem \ref{thm: spectral gap} reduces to showing
  \begin{equation}\label{eq: lb for q omega}
  h_0\pm q_{\omega}\geq -\sqrt{4+\lambda^2} \quad \mbox{for all $\omega \in \Omega = \{0,1\}^{\Z}$}.
  \end{equation}

  If $T:\Omega \to \Omega$ denotes the bijection defined by
  \begin{equation} \label{eq:flipmap}
  (T \omega)_k = -\omega_k +1,
  \end{equation}
  i.e.\ the $0$-$1$-flip map, then
  \begin{equation} \label{eq:flippot}
  q_{T\omega} = - q_{\omega}.
  \end{equation}
  Thus it suffices to establish (\ref{eq: lb for q omega}) for $h_0+q_{\omega}$. We do this by showing the existence of a positive function $\psi_{\omega}$, not necessarily belonging to $\ell^2(\mathbb{Z})$, which satisfies the difference equation
  \begin{equation}\label{eq: difference equation}
  (h_0+q_{\omega})\psi_{\omega}=-\sqrt{4+\lambda^2}\psi_{\omega}.
  \end{equation}
  This implies (\ref{eq: lb for q omega}) by Theorem~\ref{thm: Allegretto-Piepenbrink}.

Let
\[ z_{\pm}(\lambda) := \frac{\sqrt{4+\lambda^2} \pm \lambda}{2} \]
be the two distinct solutions of $z^2-z\sqrt{4+\lambda^2} +1=0$. Note that both are positive and $z_+(\lambda) z_-(\lambda)=1$. We define $\psi_{\omega}$ explicitly by
\begin{equation} \label{eq:defpsi}
\psi_{\omega}(k) = \left\{ \begin{array}{ll} z_+(\lambda)^{\sum_{j=1}^k (2\omega_j-1)}, & \mbox{if $k>0$}, \\ 1, & \mbox{if $k=0$}, \\ z_+(\lambda) ^{-\sum_{j=k+1}^0 (2\omega_j -1)}, & \mbox{if $k<0$}. \end{array} \right.
\end{equation}
Clearly, for all $k\in \Z$, $\psi_{\omega}(k)>0$ and
\[ \psi_{\omega}(k+1) = \psi_{\omega}(k) z_+(\lambda)^{2\omega_{k+1}-1}, \quad \psi_{\omega}(k-1) = \psi_{\omega}(k) z_+(\lambda)^{-(2\omega_k-1)}.\]
Thus
\begin{eqnarray*}
\lefteqn{ -\psi_{\omega}(k-1) + q_{\omega}(k) \psi_{\omega}(k) - \psi_{\omega}(k+1)} \\
& = & \psi_{\omega}(k) \left( -z_+(\lambda)^{2\omega_{k+1}-1} + \lambda (\omega_{k+1}-\omega_k) - z_+(\lambda)^{-(2\omega_k-1)} \right) \\
& = & -\sqrt{4+\lambda^2} \psi_{\omega}(k),
\end{eqnarray*}
where the latter is easily verified separately for the four cases $\omega_k \in \{0,1\}$, $\omega_{k+1} \in \{0,1\}$. This completes the proof of Theorem~\ref{thm: spectral gap}.

\end{proof}

\subsection{Density of States}

 \begin{proof}[Proof of Theorem \ref{thm: IDS symmetry}]
 Throughout the proof, we make use of the discrete Dirichlet Laplacian which on a set $\Lambda \subset \mathbb{Z}^d$ is defined by
 \begin{displaymath}
 h_{0,\Lambda}^D=h_{0,\Lambda}+N_{\Lambda},
 \end{displaymath}
 where $N_{\Lambda}$ is the edge counting operator defined in (\ref{eq: edge counting operator}).

 If $t\leq 0$, then (\ref{eq: symmetry of IDS}) is trivial.  Thus, let $t>0$.  For any displacement configuration $\omega$,  using the unitary involution (\ref{eq: unitary involution}) we have the unitary equivalence
 \begin{displaymath}
h_{\Lambda(L)}^D(\omega):= h_{0,\Lambda(L)}^D+\lambda V_{\omega}|_{\Lambda(L)} \cong -h_{0,\Lambda(L)}^N+\lambda V_{\omega}|_{\Lambda(L)}=:-H_{\Lambda(L)}^N(\omega).
 \end{displaymath}

 Thus
 \begin{eqnarray}
 \#\{k\in \mathbb{N}:E_k(h_{\Lambda(L)}^D(\omega))<E_-(\lambda)+t \}&=&\# \{k\in \mathbb{Z}:E_k(H_{\Lambda(L)}^N(\omega))> -E_-(\lambda)-t \}\nonumber\\
 &=&|\Lambda(L)|-n(H_{\Lambda(L)}^N(\omega),-E_-(\lambda)-t)\nonumber\\
 &=&|\Lambda(L)|-n(H_{\Lambda(L)}^N(\omega)+\lambda \mathbb{I}|_{\Lambda(L)},{}\nonumber\\
 &&{}{}-E_-(\lambda)-t+\lambda)\nonumber\\
 &=&|\Lambda(L)|-n(H_{\Lambda(L)}^N(\omega)+\lambda \mathbb{I}|_{\Lambda(L)}),E_+(\lambda)-t)\nonumber\\
 &=&|\Lambda(L)|-n(h_{\Lambda(L)}^N(T\omega),E_+(\lambda)-t),\label{eq: app IDS}
 \end{eqnarray}
 where, as before, in the last line $h^N=h_0^N+\lambda V$. Also, $n(A,E)$ denotes the number of eigenvalues of an operator $A$ which are less than or equal to $E$ and $T$ is the bijection on $\Omega$ defined by (\ref{eq:flipmap}).

If $p=1/2$, then $T$ is measure preserving on $\Omega$,
 \begin{equation}\label{eq: measure preserving}
\mathbb{P}(T^{-1}A)=\mathbb{P}(A)
\end{equation}
for all measurable sets $A\subset \Omega$, as it is induced by the measure preserving mapping $a\mapsto -a+1$ on $\{0,1\}$ with symmetric Bernoulli measure.

Dividing by $|\Lambda(L)|$ and taking expectations in (\ref{eq: app IDS}) gives
\begin{equation}\label{eq: introduce g}
\frac{1}{|\Lambda(L)|} \mathbb{E}(|\{k\in \mathbb{N}:E_k(h_{\Lambda(L)}^D(\omega))<E_-(\lambda)+t \}|)=1-\frac{1}{|\Lambda(L)|}\mathbb{E}(g(T\omega)),
\end{equation}
where $g(\omega)=n(h_{\Lambda(L)}^N(\omega),E_+(\lambda)-t)$.  Since $T$ is measure preserving, we have $\mathbb{E}(g(T\omega))=\mathbb{E}(g(\omega))$.  Now (\ref{eq: introduce g}) immediately gives
\begin{equation}\nonumber
\lim_{L\rightarrow \infty} \frac{1}{|\Lambda(L)|} \mathbb{E}(|\{k\in \mathbb{N}:E_k(h_{\Lambda(L)}^D(\omega))<E_-(\lambda)+t \}|)=1-\lim_{L\rightarrow \infty} \frac{1}{|\Lambda(L)|}\mathbb{E}(g(\omega)).
\end{equation}
This gives (\ref{eq: symmetry of IDS}) due to the independence of the IDS on the boundary condition.
 \end{proof}

 \begin{proof}[Proof of Theorem \ref{thm: DOS blowup}]
 We give full details on how to do the proof at $E_-(\lambda)$ and only give an outline of the modifications for proving the result at $G_+(\lambda)$.

\textit{Lower Bound at $E_-(\lambda)$}:  Here we closely follow an argument developed in a similar context for the {\it continuum} random displacement model in \cite{Baker} and \cite{BLS2}.

Let $h_{0,\Lambda(L)}^{\bullet}(\omega):=h_{0,\Lambda(L)}^{\bullet}+\lambda V_{\omega}|_{\Lambda(L)}$ for $\bullet\in \{D,N,P\}$.  We will use the standard a priori bound
\begin{equation}
  N_{\lambda}(E) \geq \frac{1}{|\Lambda(L)|}\mathbb{P}\bigg(E_1(h_{\Lambda(L)}^D(\omega))<E\bigg), \label{eq:a priori lb}
\end{equation}
see, for example, \cite[(6.15)]{Kirsch}. We will show that $E_1(h_{\Lambda(L)}^D(\omega))<E$ by constructing a test function $\Psi_{\omega}$ with Rayleigh quotient $\langle \Psi_{\omega}, h_{\Lambda(L)}^D \Psi_{\omega} \rangle/ \|\Psi_{\omega}\|^2 <E$.

To this end, let $\phi_1$ denote the strictly positive ground state of $h_{\Lambda_1}^N+\lambda \delta_2$ normalized so that $\phi_1(1)=1$, and let $\phi_{-1}$ denote the strictly positive ground state of $h_{\Lambda_1}^N+\lambda \delta_1$ normalized so that $\phi_{-1}(1)=1$.  Using uniqueness of the positive ground state, we have $\phi_{-1}(2)=\phi_1(2)^{-1}$.  In fact, $\phi_{-1}^f$ given by $\phi_{-1}^f(n)=\phi_1(2-n+1)$ is a positive ground state of $h_{\Lambda_1}^N+\lambda \delta_1$.  For ease of notation, we denote $r:=\phi_1(2)$. We know by Lemma~\ref{lemma: unequal at endpoints} that $r\not = 1$ and will now assume that $r>1$. If $0<r<1$, we can do the following construction from ``right to left'', choosing the test function $\Psi_{\omega}$ such that $\Psi_{\omega}(2L)=1$.

Given $\omega$, we construct $\Psi_{\omega}$ by concatenating rescaled versions of $\phi_{-1}$ and $\phi_1$.  The test function $\Psi_{\omega}$ is defined piecewise on $\Lambda_k$, $1\leq k \leq L$, by
\begin{equation}\label{eq:test function}
  \Psi_{\omega}(n)=r^{\sum_{i=1}^{k-1}\alpha_i(\omega)}\phi_{\alpha_k(\omega)}(n-2(k-1)) \qquad n\in \Lambda_k,
\end{equation}
where $\alpha_k(\omega)=-1$ if $\omega_k=0$ and $\alpha_k(\omega)=1$ if $\omega_k=1$. This choice of rescaling guarantees that $\Psi_{\omega}(2k)= \Psi_{\omega}(2k+1)$ for all $k$ and thus, by the properties of Neumann boundary conditions on all $\Lambda_k$, that $h_{\Lambda(L)}^N(\omega)\Psi_{\omega}=E_0(1)\Psi_{\omega}$.  However, as we have shown in the proof of Theorem~\ref{thm: omega star}, $E_0(1) = E_- = \min \sigma(h_0 +V_{\omega^*})$. Thus we have $h_{0,\Lambda(L)}^N\Psi_{\omega}=E_-\Psi_{\omega}$ and we may write
\begin{equation}
  \frac{\bra \Psi_{\omega},h_{\Lambda(L)}^D(\omega)\Psi_{\omega}\ket}{\bra \Psi_{\omega},\Psi_{\omega}\ket}-E_- =  \frac{\bra\Psi_{\omega}, h_{\Lambda(L)}^D(\omega)\Psi_{\omega}-h_{0,\Lambda(L)}^N(\omega)\Psi_{\omega}\ket}{\bra \Psi_{\omega},\Psi_{\omega}\ket}.\label{eq:numerator}
\end{equation}

The Dirichlet and Neumann operators only differ at the endpoints of $\Lambda(L)$, so most terms in the numerator of (\ref{eq:numerator}) cancel; what rests is

\begin{equation}
  \frac{\bra \Psi_{\omega},h_{\Lambda(L)}^D(\omega)\Psi_{\omega}\ket}{\bra \Psi_{\omega},\Psi_{\omega}\ket}-E_-=\frac{2(1+\Psi_{\omega}(2L)^2)}{\bra \Psi_{\omega},\Psi_{\omega}\ket}.
\end{equation}

With (\ref{eq:a priori lb}), we have
\begin{eqnarray*}
N_{\lambda}(E) & \geq & \frac{1}{2L} \mathbb{P} \left( \frac{\langle \Psi_{\omega}, h_{\Lambda(L)}^D(\omega) \Psi_{\omega} \rangle}{\langle \Psi_{\omega}, \Psi_{\omega} \rangle} < E \right) \\
& = &  \frac{1}{2L}\mathbb{P}\bigg(\frac{2(1+\Psi_{\omega}(2L)^2)}{\bra \Psi_{\omega},\Psi_{\omega}\ket}<E-E_- \bigg) \\
& \geq & \frac{1}{2L}\mathbb{P}\bigg(\frac{2(1+\Psi_{\omega}(2L)^2)}{\sum_{k=1}^L\Psi_{\omega}(2k)^2}<E-E_- \bigg).
\end{eqnarray*}

It follows from the definition of $\Psi_{\omega}$ that
\begin{displaymath}
\alpha_j(\omega)=\frac{\log\Psi_{\omega}(2j)/\Psi_{\omega}(2j-1)}{\log r}=\frac{\log \Psi_{\omega}(2j+1)/\Psi_{\omega}(2j-1)}{\log r}
\end{displaymath}
(where we use that $\Psi_{\omega}(2j)=\Psi_{\omega}(2j+1)$) and
\begin{equation}\label{eq:value at endpoint}
\Psi_{\omega}(2j)^2=e^{2S_j\log r},
\end{equation}
where $S_j:=\alpha_1+\alpha_2+\cdots+\alpha_j$.  As $\mathbb{P}(\alpha_k(\omega)=1)=\mathbb{P}(\omega_k=1)=1/2$ and $\mathbb{P}(\alpha_k(\omega)=-1)=\mathbb{P}(\omega_k=0)=1/2$, the $\alpha_j(\omega)$ are independent symmetric Bernoulli random variables with values in $\{\pm1\}$.  Therefore, the process $S_j$, $j=1,2,\ldots$, is a simple, symmetric random walk.  If $Y:=\max_{i=1,\ldots,L}S_i$, then it is a consequence of the reflection principle for symmetric random walks (\cite{Feller}) that
\begin{equation}\label{eqofprobabilities}
\mathbb{P}(Y\geq \sqrt{L} | S_L\leq0)=\mathbb{P}(S_L\geq 2\sqrt{L}).
\end{equation}
The latter converges to
\begin{equation}\label{integrallimit}
\pi^{-1/2}\int_{2}^{\infty} e^{-y^2/2}\,dy>0
\end{equation}
as $L\rightarrow \infty$ by the central limit theorem.

Denote by $A_L$ the event $\{\omega|Y\geq \sqrt{L} \ \textrm{and} \ S_L\leq 0\}$.  If $Y\geq \sqrt{L}$, then $\sum_{j=1}^L \Psi_{\omega}(2j)^2\geq e^{2\sqrt{L}\log r}$.  The condition $S_L\leq 0$ means that equal or more single site potentials sit at the left than sit at the right on $\Lambda(L)$.  As $r>1$, it is clear from (\ref{eq:value at endpoint}) that $\Psi_{\omega}(2L)^2\leq1$ and we have

\begin{eqnarray}
  \mathbb{P}\bigg(E_1(h_{\Lambda(L)}^D(\omega))<E \bigg)&\geq&\mathbb{P}\bigg(\frac{2(1+\Psi_{\omega}(2L)^2)}{\sum_{j=1}^L\Psi_{\omega}(2j)^2}<E-E_-|A_L \bigg)\mathbb{P}(A_L)\nonumber\\
  &\geq&\mathbb{P}\bigg(4e^{-2\sqrt{L}\log r}<E-E_-|A_L \bigg)\mathbb{P}(A_L)\nonumber\\
  &=&\mathbb{P}(A_L)\nonumber\\
  &=&\underbrace{\mathbb{P}(Y\geq \sqrt{L}|S_L\leq0)}_{\geq C_0}\cdot \underbrace{\mathbb{P}(S_L\leq 0)}_{\geq \frac{1}{2}}\geq \frac{1}{2}C_0>0\label{eq:prob lb}
\end{eqnarray}
if $4e^{-2\sqrt{L}\log r}<E-E_-$ and $L>L_0$ (where we have used (\ref{eqofprobabilities}) and (\ref{integrallimit})).  If $E$ is so close to $E_-$ that
\begin{displaymath}
 E-E_-<4e^{-2\sqrt{L_0}\log r},
\end{displaymath}
then we can choose $L\geq L_0$ for which
\begin{equation}
4e^{-2\sqrt{L-1}\log r}\geq E-E_-\geq 4e^{-2\sqrt{L}\log r}.
\end{equation}
Therefore there are constants $C_1,C_2>0$ such that
\begin{equation}\label{c1c2}
C_1L\leq\bigg[\frac{\log \frac{1}{4}(E-E_-)}{\log r} \bigg]^2 \leq C_2L.
\end{equation}

From (\ref{eq:prob lb}) we get $N_{\lambda}(E)\geq \frac{C_0}{2}\frac{1}{L}$, so that
\begin{displaymath}
N_{\lambda}(E)\geq\frac{C_0}{2C_1}\bigg(\frac{\log r}{\log \frac{1}{4}(E-E_-)} \bigg)^2\geq \frac{C}{\log^2 (E-E_-)}
\end{displaymath}
for $E-E_-$ sufficiently small.

\textit{Discussion of Modifications for Lower Bound at $G_+(\lambda)$}:  We want to consider a lower bound for $N_{\lambda}(G_+(\lambda)+\epsilon)$, with $\epsilon>0$.  The idea is to express the integrated density of states $N_{\lambda}$ at the gap edge $G_+(\lambda)$ in terms of the integrated density of states of the operators $h_0\pm q_{\omega}$ from the proof of Theorem~\ref{thm: spectral gap}, call them $N_{\pm}$, at $-\sqrt{4+\lambda^2}$.  One then estimates $N_{\pm}$ from below.

Let $L\in \mathbb{N}$ and define
\begin{equation}\label{eq: H restricted}
H_{\omega,\lambda}^L:=\bigg(h_{[1,2L]}^P(\omega)-\frac{\lambda}{2} \bigg)^2-\bigg(2+\frac{\lambda^2}{4} \bigg).
\end{equation}

It follows from Theorems~\ref{thm: deterministic spectrum} and \ref{thm: spectral gap} that $h_{[1,2L]}^P(\omega)$ has no eigenvalues in the interval $(G_-(\lambda),G_+(\lambda))$.  Thus, $H_{\omega,\lambda}^L\geq -\sqrt{4+\lambda^2}$.  Moreover, via the same construction as in Section~\ref{sec:asspectrum}, $H_{\omega,\lambda}^L$ is unitarily equivalent to the direct sum
\begin{displaymath}
H_{\omega,\lambda}^L\cong \left( h_{0,[1,2L]}^P-q_{\omega}|_{[1,2L]} \right) \oplus \left( h_{0,[1,2L]}^P+q_{\omega}|_{[1,2L]} \right)=:J_L^-\oplus J_L^+.
\end{displaymath}

Therefore, setting $\eta(\epsilon) := 2\epsilon (G_+(\lambda)-\lambda/2)+\epsilon^2$,
\begin{eqnarray*}
N_{\lambda}(G_+(\lambda)+\epsilon)-\frac{1}{2}&=&\frac{1}{2}\bigg(N_{\lambda}(G_+(\lambda)+\epsilon)-N_{\lambda}(G_-(\lambda)-\epsilon) \bigg) \\
&=&\lim_{L\rightarrow \infty}\frac{1}{4L}\mathbb{E}\bigg(\#\big\{k:E_k(h_{[1,2L]}^P(\omega))\in [G_-(\lambda)-\epsilon,G_+(\lambda)+\epsilon]\big\} \bigg) \\
&=&\lim_{L\rightarrow \infty}\frac{1}{4L}\mathbb{E}\bigg(\#\big\{k:E_k\big((h_{[1,2L]}^P(\omega)-\frac{\lambda}{2})^2\big)\leq \big(G_+(\lambda)-\frac{\lambda}{2}+\epsilon\big)^2\big\} \bigg) \\
&=&\lim_{L\rightarrow \infty}\frac{1}{4L}\mathbb{E}\bigg(\#\big\{k:E_k(J_L^-\oplus J_L^+)\leq -\sqrt{4+\lambda^2}+\eta(\epsilon)\big\} \bigg) \\
&=&\lim_{L\rightarrow \infty}\frac{1}{4L}\mathbb{E}\bigg(\#\big\{k:E_k(J_L^- )\leq -\sqrt{4+\lambda^2}+\eta(\epsilon)\big\}+\nonumber\\
&& +\#\big\{k:E_k(J_L^+ )\leq -\sqrt{4+\lambda^2}+\eta(\epsilon)\big\} \bigg) \\
& = & \frac{1}{2} N_-(-\sqrt{4+\lambda^2}+\eta(\epsilon)).
\end{eqnarray*}
The last equality uses that $N_-=N_+$, which follows from the fact that $T$ defined by (\ref{eq:flipmap}) is measure-preserving on $\Omega$ and satisfies $q_{T\omega}=-q_{\omega}$.
 In particular,
\begin{equation} \label{eq:IDSlowerbound}
N_{\lambda}(G_+(\lambda)+\epsilon) - \frac{1}{2} \ge \frac{1}{2} N_-\left(-\sqrt{4+\lambda^2} + 2\epsilon (G_+(\lambda) - \frac{\lambda}{2})\right).
\end{equation}

One can prove that $N_-$ has a $1/\log^2$-singularity near $-\sqrt{4+\lambda^2}$ with essentially the same techniques used for $N_{\lambda}$ near $E_-(\lambda)$.  For $N_-$ one chooses a different trial function $\Psi_{\omega}$. The appropriate choice is $\Psi_{\omega}:=\psi_{\omega}|_{[1,L]}$, where $\psi_{\omega}$ is the positive function defined by (\ref{eq:defpsi}) in the proof of Theorem \ref{thm: spectral gap}.  With this choice, $\Psi_{\omega}$ is, up to an error term, an eigenfunction of $h^D_{0,\Lambda(L)}-q_{\omega}|_{\Lambda(L)}$ with eigenvalue $-\sqrt{4+\lambda^2}$.  More precisely,
\begin{eqnarray*}
\lefteqn{\langle \Psi_{\omega}, (h_{0,[1,L]}^D - q_{\omega}|_{[1,L]}) \Psi_{\omega} \rangle} \nonumber \\ & = & -\sqrt{4+\lambda^2} \,\|\Psi_{\omega}\|^2 + \Psi_{\omega}(1)^2 - \Psi_{\omega}(1) \Psi_{\omega}(2) - \Psi_{\omega}(L) \Psi_{\omega}(L-1) + \Psi_{\omega}(L)^2.
\end{eqnarray*}

Thus we have the lower bound
\begin{eqnarray*}
\lefteqn{N_-(E) \:\geq \: \frac{1}{L}\mathbb{P}\big(E_1(h_{0,[1,L]}^D-q_{\omega}|_{[1,L]})<E \big)} \\
&\geq&\frac{1}{L}\mathbb{P} \left( \frac{ \langle \Psi_{\omega}, (h^D_{0,[1,L]} - q_{\omega}|_{[1,L]}) \Psi_{\omega} \rangle}{\|\Psi_{\omega}\|^2} + \sqrt{4+\lambda^2} < E +\sqrt{4+\lambda^2} \right) \\
& \ge & \frac{1}{L} \mathbb{P} \left( \frac{\Psi_{\omega}(1)^2- \Psi_{\omega}(1) \Psi_{\omega}(2) - \Psi_{\omega}(L) \Psi_{\omega}(L-1) + \Psi_{\omega}(L)^2}{\|\Psi_{\omega}\|^2} < E+\sqrt{4+\lambda^2} \right).
\end{eqnarray*}

From here on the proof is completed in a similar fashion to what was done above for $E_-(\lambda)$. Key is the simple multiplicative structure of $\psi_{\omega}$ as defined in (\ref{eq:defpsi}), which leads to considering the symmetric random walk defined by the Bernoulli variables
\[ \alpha_k := \frac{\log \Psi_{\omega}(k+1)/ \Psi_{\omega}(k)}{\log z_+(\lambda)} = 2\omega_{k+1}-1.\]
If $\lambda>0$, then $z_+(\lambda)>1$ and the proof goes through with the same argument as above. If $\lambda<0$, then $0<z_+(\lambda)<1$ and one can work from ``right to left'', similar to what was indicated for the case $0<r<1$ above.

 \end{proof}

\section{Concluding Remarks} \label{sec:discussion} \label{sec:conclusion}

(i) Our results for $d>1$ fall short of what was proven in \cite{BLS1} and \cite{BLS2} for the {\it continuum} displacement model in several respects:

\begin{itemize}

\item In part (i) of Theorem~\ref{thm: bubbles} we have required sign-definiteness of $q$. The corresponding result for the continuum from \cite{BLS1} only requires that $E_0(a)$ does not vanish identically in $a$. In other words, in the continuum the multi-dimensional analogue of part (ii) of Theorem~\ref{thm: bubbles} holds (noting that $0$ is the spectral minimum of the continuum Neumann Laplacian, while $-2$ is the spectral minimum of the discrete one-dimensional Neumann Laplacian).

\item We have yet to prove any uniqueness results for the set of periodic configurations $\omega$ with $\min \sigma(h_{\omega})= E_{min}$ in the case of $d\geq 2$, but we conjecture the following analogue of a result which was shown to hold for the continuum case in \cite{BLS2}:

\begin{conjecture}
If $d\geq2$ and $q\neq0$ is sign-definite, then $\omega^{\ast}$ is (up to translation) the unique periodic configuration with $\min \sigma(h_{\omega^*}) = E_{min}$ and $\max \sigma(h_{\omega^*}) = E_{max}$.
\end{conjecture}

\end{itemize}

Both of these shortcomings of our results are due to the lack of unique continuation properties of the discrete Schr\"odinger equation, which were used in this context in the continuum in \cite{BLS1, BLS2}.

(ii) It has recently been shown in \cite{KLNS} that the continuum multi-dimensional random displacement model, under symmetry assumptions similar to the ones used here, exhibits Anderson localization at energies near the bottom of the almost sure spectrum. Proving this also in the discrete case remains a challenging problem. A uniqueness result as conjectured in (i) above would allow to prove Lifshitz tails for the IDS with a strategy developed in \cite{KN} (and used in \cite{KLNS}). However, due to the ``discrete'' nature of the randomness we can not expect a Wegner estimate to hold for the discrete displacement model. The new type of multiscale analysis recently developed in \cite{BK} to circumvent this problem for the continuum Bernoulli-Anderson model does not work on the lattice, again due to the lack of unique continuation properties.

(iii) For the one-dimensional Bernoulli displacement model we have no result like Corollary~\ref{cor: lambda less than 2} if $|\lambda|>2$, but our conjecture is

\begin{conjecture}\label{conj: as spec for all lambda}
For any $\lambda \neq 0$,
\begin{equation} \label{eq:conj2}
\Sigma_{\lambda}=\sigma(h_{\omega^{\ast},\lambda})\cup \sigma(h_{\omega^1,\lambda}),
\end{equation}
where $\omega^1$ is the displacement configuration with components $\omega^1_k=1$ for all $k\in \mathbb{Z}$.
\end{conjecture}

If Conjecture \ref{conj: as spec for all lambda} is true, it would mean that $\Sigma_{\lambda}$ consists of exactly six bands separated by five non-vanishing spectral gaps.  Other than the fact that this would be a natural generalization of what we can prove for $|\lambda|\le 2$ (where Corollary~\ref{cor: lambda less than 2} followed as a consequence of (\ref{eq:conj2})), we have numerical evidence for this conjecture which was provided to us by Jeff Baker. Figure~\ref{conj: as spec for all lambda} shows numerical plots of the density of states (through histograms of the eigenvalues of finite volume restrictions averaged over multiple sample configurations) as well as the IDS for $\lambda=3$. The same pattern arises numerically for other values of $\lambda>2$.

    \begin{figure}[h]
\centering
  \includegraphics[width=\textwidth]{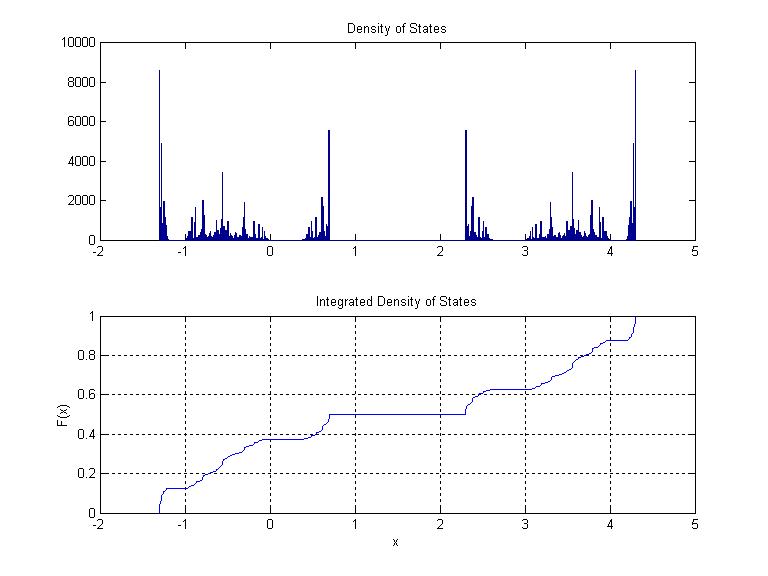}
  \caption{The density of states and integrated density of states for the one-dimensional BDM with $\lambda=3$ and $p=1/2$.}
  \label{conj: as spec for all lambda}
\end{figure}

(iv) Note that Figure~\ref{conj: as spec for all lambda} also illustrates the $1/\log^2$-lower bounds for the growth of the IDS found at the band edges $E_{\pm}(\lambda)$ and $G_{\pm}(\lambda)$ in Theorem \ref{thm: DOS blowup}.  One of the reasons which make this bound interesting is that it was proven in \cite{CraigSimon} that for general ergodic Jacobi matrices the IDS satisfies an upper bound of the form $|N(E)-N(E')| \le C/ |\log|E-E'||$ locally at all energies, i.e.\ that $N(E)$ is {\it log-H\"older continuous}.
It would be interesting to find out if our lower bound is sharp, i.e.\ if there is also an upper bound of type $1/\log^2$.

Known examples showing the optimality of the $1/|\log|$ upper bound were provided in \cite{CraigSimon} and, more recently, in \cite{GanKrueger}. These examples are in terms of quasi-periodic and limit-periodic potentials. As opposed to our example, these are non-random in the sense of having long-range spatial correlations.

(v) The eight additional band edges which would exist for $|\lambda|>2$ if Conjecture~\ref{conj: as spec for all lambda} were to hold show different numerical properties in Figure~\ref{conj: as spec for all lambda}. This is most drastically apparent in the density of states plot, where delta-peaks appear at the edges addressed by Theorem \ref{thm: DOS blowup}, but not at the other eight edges. Based on the numerics we state

\begin{conjecture}\label{conj: lifshits tails}
If $|\lambda|>2$, then the IDS has thin tails (e.g.\ of Lifshits type) at the eight other band edges of $\Sigma_{\lambda}$.
\end{conjecture}

\end{document}